%% file: EnsuringLiveness.tex
\documentclass[hidelinks]{eptcs}
\usepackage{breakurl}
\usepackage{latexsym}
\usepackage{amsmath}
\usepackage{wrapfig}
\usepackage{color}
\usepackage{stmaryrd}

\DeclareSymbolFont{frenchscript}{OMS}{ztmcm}{m}{n}
\DeclareMathSymbol{\A}{\mathord}{frenchscript}{65}    
\DeclareMathSymbol{\Ce}{\mathord}{frenchscript}{67}   
\DeclareMathSymbol{\F}{\mathord}{frenchscript}{70}    
\DeclareMathSymbol{\K}{\mathord}{frenchscript}{75}    
\DeclareMathSymbol{\Pow}{\mathord}{frenchscript}{80}  
\DeclareMathSymbol{\Tk}{\mathord}{frenchscript}{84}   
\DeclareMathAlphabet{\mathbbm}{U}{bbm}{m}{n}          
\newcommand{\IT}{\mathbbm{T}}                         
\makeatletter
\def\comesfrom{\@transition\leftarrowfill}
\def\goesto{\@transition\rightarrowfill}
\def\ngoesto{\@transition\nrightarrowfill}
\def\Goesto{\@transition\Rightarrowfill}
\def\nGoesto{\@transition\nRightarrowfill}
\def\xmapsto{\@transition\mapstofill}
\def\nxmapsto{\@transition\nmapstofill}
\def\@transition#1{\@@transition{#1}}
\newbox\@transbox
\newbox\@arrowbox
\newbox\@downbox
\def\@@transition#1#2%
   {\setbox\@transbox\hbox
      {\vrule height 1.5ex depth .8ex width 0ex\hskip0.25em$\scriptstyle#2$\hskip0.25em}
   \ifdim\wd\@transbox<1.5em
      \setbox\@transbox\hbox to 1.5em{\hfil\box\@transbox\hfil}\fi
   \setbox\@arrowbox\hbox to \wd\@transbox{#1}
   \ht\@arrowbox\z@\dp\@arrowbox\z@
   \setbox\@transbox\hbox{$\mathop{\box\@arrowbox}\limits^{\box\@transbox}$}
   \dp\@transbox\z@\ht\@transbox 12pt
   \mathrel{\box\@transbox}}
\def\nrightarrowfill{$\m@th\mathord-\mkern-6mu%
  \cleaders\hbox{$\mkern-2mu\mathord-\mkern-2mu$}\hfill
  \mkern-6mu\mathord\not\mkern-2mu\mathord\rightarrow$}
\def\Rightarrowfill{$\m@th\mathord=\mkern-6mu%
  \cleaders\hbox{$\mkern-2mu\mathord=\mkern-2mu$}\hfill
  \mkern-6mu\mathord\Rightarrow$}
\def\nRightarrowfill{$\m@th\mathord=\mkern-6mu%
  \cleaders\hbox{$\mkern-2mu\mathord=\mkern-2mu$}\hfill
  \mkern-6mu\mathord\not\mathord\Rightarrow$}
\def\mapstofill{$\m@th\mathord\mapstochar\mathord-\mkern-6mu%
  \cleaders\hbox{$\mkern-2mu\mathord-\mkern-2mu$}\hfill
  \mkern-6mu\mathord\rightarrow$}
\def\nmapstofill{$\m@th\mathord\mapstochar\mathord-\mkern-6mu%
  \cleaders\hbox{$\mkern-2mu\mathord-\mkern-2mu$}\hfill
  \mkern-6mu\mathord\not\mkern-2mu\mathord\rightarrow$}
\makeatother 
\newcommand{\goto}[2][]{\mathrel{\goesto{~#2{\color{red}\;,\;#1}~}}}        
\newtheorem{proposition}{Proposition}
\newtheorem{open}{Open Problem}
\newenvironment{proof}{\begin{trivlist} \item[\hspace{\labelsep}\bf Proof:]}{\hfill $\Box$\end{trivlist}}
\newtheorem{defi}{Definition}
\newenvironment{definition}[1]{\begin{defi} \rm \label{df:#1} }{\end{defi}}
\newenvironment{definitionA}[2]{\begin{defi}[#1] \rm \label{df:#2} }{\end{defi}}
\newcommand{\df}[1]{Definition~\ref{df:#1}}
\newcommand{\weg}[1]{}                                
\newcommand{\plat}[1]{\raisebox{0pt}[0pt][0pt]{#1}}   
\newcommand{\den}[1]{\llbracket#1\rrbracket}          
\newcommand{\dcup}{\stackrel{\mbox{\huge .}}{\cup}}   
\newcommand{\cT}{{\rm T}}                             
\newcommand{\J}{{\rm J}}                              
\newcommand{\Tr}{\textit{Tr}}                         
\newcommand{\source}{\textit{source\/}}               
\newcommand{\target}{\textit{target\/}}               
\newcommand{\comp}{\textit{comp\/}}                   
\newcommand{\npc}{\textit{npc\/}}                     
\newcommand{\afc}{\textit{afc\/}}                     
\newcommand{\Left}{\textsc{l}}                        
\newcommand{\R}{\textsc{r}}                           
\newcommand{\aconc}{\mathrel{\mbox{$\smile\hspace{-.95ex}\raisebox{2.5pt}{$\scriptscriptstyle\bullet$}$}}}
\newcommand{\naconc}{\mathrel{\mbox{$\,\not\!\smile\hspace{-.95ex}\raisebox{2.5pt}{$\scriptscriptstyle\bullet$}$}}}
\newcommand{\nconc}{\,\not\!\smile}                   

\def\titlerunning{Ensuring Liveness Properties of Distributed Systems}
\def\authorrunning{R.J. van Glabbeek}
\title{\titlerunning:\linebreak[1] Open Problems\vspace{-2pt}}
\author{Rob van Glabbeek
\institute{Data61, CSIRO, Sydney, Australia}
\institute{School of Computer Science and Engineering,
University of New South Wales, Sydney, Australia}
\email{rvg@cs.stanford.edu}
}
\providecommand{\publicationstatus}{July 2019. To appear in
Ilaria Castellani, Pedro D'Argenio, Mohammad Reza Mousavi \& Ana Sokolova, eds.:\\
``Open Problems in Concurrency Theory'', a special issue of the
Journal of Logical and Algebraic Methods in Programming.}
\begin{document}
\maketitle

\begin{abstract}
Often fairness assumptions need to be made in order to establish
liveness properties of distributed systems, but in many situations
they lead to false conclusions.

This document presents a research agenda aiming at laying the
foundations of a theory of concurrency that is equipped to ensure
liveness properties of distributed systems without making fairness
assumptions.  This theory will encompass process algebra, temporal
logic and semantic models.
The agenda also includes the development of a methodology and tools
that allow successful application of this theory to the specification,
analysis and verification of realistic distributed systems.

Contemporary process algebras and temporal logics fail to make
distinctions between systems of which one has a crucial liveness
property and the other does not, at least when assuming \emph{justness},
a strong progress property, but not assuming fairness.
Setting up an alternative framework involves giving up on identifying
strongly bisimilar systems, inventing new induction principles,
developing new axiomatic bases for process algebras and new congruence
formats for operational semantics, and creating matching treatments of time
and probability.

Even simple systems like fair schedulers or mutual exclusion protocols
cannot be accurately specified in standard process algebras (or Petri
nets) in the absence of fairness assumptions.  Hence the work involves
the study of adequate language or model extensions, and their
expressive power.
\end{abstract}

\section{State-of-the-art}\label{state-of-the-art}
\advance\textheight 13.6pt
\advance\textheight 13.6pt

\subsection{Specification, analysis and verification of distributed systems}

At an increasing rate, humanity is creating distributed systems
through hardware and software---systems consisting of multiple
components that interact with each other through message passing or
other synchronisation mechanisms.  Examples are distributed databases,
communication networks, operating systems, industrial control systems,
etc. Many of these systems are hard to understand, yet vitally
important. Therefore, significant effort needs to be made to
ensure their correct working.

Formal methods are an indispensable tool towards that end.
They consist of specification formalisms to unambiguously capture the
intended requirements and behaviour of a system under consideration, tools and
analysis methods to study and reason about vital properties of the
system, and mathematically rigorous methods to verify that (a) a
system specification ensures the required properties, and (b) an
implementation meets the specification.

The standard alternative to formal specification formalisms are descriptions
in English, or other natural languages, that try to specify the
requirements and intended workings of a system. History has shown,
almost without exception, that such descriptions are riddled with
ambiguities, contradictions and under-specification. 
Formalisation of such a description---regardless in which
formalism---is the key to elimination of these holes.

A formal specification of a distributed system typically comes in (at
least) two parts.

One part formulates the \emph{requirements}
imposed on the system as a list of properties the system should
have. Amongst the formalisms to specify such requirements are temporal
logics like Linear-time Temporal Logic (LTL) \cite{Pn77} or Computation Tree Logic
(CTL) \cite{EC82}.  Amongst others, they can specify \emph{safety properties},
saying that something bad will never happen, and \emph{liveness properties},
saying that something good will happen eventually \cite{Lam77}.

The other part is a formal description of how the system ought to
work on an \emph{operational} (= step by step) basis, but abstracting from
implementation details. For distributed systems such accounts
typically consist of descriptions of each of the parallel components,
as well as of the communication interfaces that specify how different
components interact with each other. Languages for giving such formal
descriptions are \emph{system description languages}.  When a
system description language features constants to specify elementary
system activities, and operators (like parallel or sequential
composition) to create more complex systems out of simpler ones, it is
sometimes called a \emph{process algebra}.
Alternatively, operational system descriptions can be rendered in a
model of concurrency, such as Petri nets or labelled transition systems.
Such models are also used to describe the meaning of system
description languages.

Once such a two-tiered formalisation of a system has been provided,
there are two obvious tasks to ensure the correct working of
implementations: (a) guaranteeing that the operational system description
meets the requirements imposed on the system, and (b) ensuring that an
implementation satisfies the specification. The latter task additionally
requires a definition of what it means for an implementation to
satisfy a specification, and this definition should ensure that
any relevant correctness properties that are shown to hold for the
specification also hold for the implementation.

A third type of task is the study of other properties of the
implementation, not implied by the specification.  Examples are
measuring its execution times, when these are not part of the
specification, or its success rate, for operations for which success
cannot be guaranteed and only a best effort is made. Potentially,
these tasks call for applications of probability theory.

Traditional approaches to ensure the correct working of distributed
systems are simulation and test-bed experiments. While these are
important and valid methods for system evaluation, in particular for
quantitative performance evaluation, they have limitations in regards
to the evaluation of basic correctness properties. Experimental
evaluation is resource-intensive and time-consuming, and, even after a
very long time of evaluation, only a finite set of operational
scenarios can be considered---no general guarantee can be given about
correct system behaviour for a wide range of unpredictable deployment
scenarios.
I believe that formal methods help in this regard; they complement
simulation and test-bed experiments as methods for system evaluation
and verification, and provide stronger and more general assurances
about system properties and behaviour.

\subsection{Achievements of process algebra and related formalisms}

\emph{Process algebra} is a family of approaches to the specification,
analysis and verification of distributed systems. Its tools encompass
algebraic languages for the specification of systems (mentioned
above), algebraic laws to reason about system descriptions,
and induction principles to derive behaviours of infinite systems from
those of their finite approximations.

Many industrial size distributed systems have been successfully
specified, analysed and verified in frameworks based on process algebra--\!\!--%
examples can be found through the following links.
Major toolsets primarily based on process algebra include
\href{http://www.cs.ox.ac.uk/projects/fdr/}{FDR} \cite{GRABR14},
\href{http://cadp.inria.fr/}{CADP} \cite{GLMS11},
\href{http://www.win.tue.nl/mcrl2/}{mCRL2} \cite{GM14} and
the \href{http://www.it.uu.se/research/group/mobility/applied/psiworkbench}{Psi-Calculi Workbench} \cite{BGRV15,BJPV11}.
Key methods employed are \emph{model checking} \cite{BK08} to ensure that an
operational system description meets system requirements, and
\emph{equivalence checking} \cite{CS01} 
or \emph{refinement} \cite{Mor94}, 
to show that an operational description of an implementation is equivalent to or
at least as suitable as an operational system specification,
in the sense that it inherits all its relevant good properties.

Additional toolsets primarily based on model checking or other mathematical
techniques that explore the state spaces of distributed systems include
\href{http://spinroot.com/spin/}{SPIN} \cite{SPIN},
\href{http://www.uppaal.org/}{\sc Uppaal} \cite{BDL04},
\href{http://ltsmin.utwente.nl/}{LTSmin} \cite{KLMPBD15},
\href{http://www.prismmodelchecker.org/}{PRISM} \cite{KNP10} and
\href{http://research.microsoft.com/en-us/um/people/lamport/tla/toolbox.html}{TLA} \cite{La02}.
\vspace{-4pt}

\subsection{Liveness and fairness assumptions}\label{liveness}

One of the crucial tasks in the analysis of distributed systems is the
verification of liveness properties, saying that something good will
happen eventually. A typical example is the verification of a communication
protocol---such as the \emph{alternating bit protocol} \cite{Lyn68,BSW69}---that
ensures that a stream of messages is relayed correctly, without loss or
reordering, from a sender to a receiver, while using an unreliable
communication channel.  The protocol works by means of acknowledgements,
and resending of messages for which no acknowledgement is received.

Naturally, no protocol is able to ensure such a thing, unless one
assumes that attempts to transmit a message over the unreliable channel
will not fail in perpetuity. Such an assumption, essentially saying
that if one keeps trying something forever it will eventually succeed, is
often called a \emph{fairness} assumption.
See \cite{GH18} for a formal definition of fairness, and an overview
of notions of fairness from the literature.

Making a fairness assumption is indispensable when verifying the
alternating bit protocol. If one refuses to make such an assumption,
no such protocol can be found correct, and one misses a chance to
check the protocol logic against possible flaws that have nothing to
do with a perpetual failure of the communication channel.

For this reason, fairness assumptions are made in many process-algebraic
verification methods, and are deeply ingrained in their methodology \cite{BBK87a}.
This applies, amongst others, to the default incarnations of the
process algebras CCS \cite{Mi89}, ACP \cite{BW90,Fok00,BBR10}, the
$\pi$-calculus \cite{Mi99,SW01} and mCRL2 \cite{GM14}.
\vspace{-4pt}

\subsection{Fairness assumptions in process algebra}\label{fair abstraction}

\textsl{This section explains the last claim, and can be skipped
  by anyone unfamiliar with these process algebras.}\linebreak[3]
A typical process-algebraic verification of---say---the alternating bit protocol \cite{BK86}
starts from process algebraic descriptions of the specification as well as the implementation.
Each comes as an expression in a suitable process-algebraic specification language,
and is interpreted as a state in a labelled transition system.
The specification describes the intended behaviour of the protocol,
and does not model message loss due to unreliable communication channels.
The implementation describes the parallel composition of the sender,
the receiver, and the unreliable channels for messages and acknowledgements.
All unsuccessful communication attempts manifest themselves as cycles of
internal transitions (labelled with the unobservable action $\tau$)
in the labelled transition system that forms the semantic model of the composed
implementation. The verification consists of showing that the implementation
is semantically equivalent to the specification, using an equivalence
such as \emph{branching bisimulation equivalence} \cite{GW96} or
\emph{weak bisimulation equivalence} \cite{Mi89}. These semantic equivalences
abstract from cycles of internal actions, and this is essential for a successful verification.
The technique of proving liveness properties through abstraction from $\tau$-cycles is
called \emph{fair abstraction} \cite{BK86,BRV96}.
It amounts to applying a particular strong fairness assumption,
called \emph{full fairness} in \cite{GH18}.
In \cite{BK86} the equivalence
of specification and implementation is obtained by algebraic manipulation of
process-algebraic expressions, and here the abstraction from $\tau$-cycles is captured
by \emph{Koomen's Fair Abstraction Rule}. The soundness of this proof rule w.r.t.\
weak bisimulation equivalence is established in \cite{BBK87a}.

A popular alternative to bisimulation equivalence for relating specifications and implementations is the
\emph{must testing equivalence} of \cite{DH84}, which coincides with the \emph{failures equivalence} of CSP
\cite{BHR84,Ho85,Ros97}. It does not embody fairness assumptions.
However, a variant that incorporates fair abstraction has been proposed in \cite{Vo92,BRV95,NC95}.
\pagebreak

\section{The dangers of assuming fairness}\label{dangers}

Using a fairness assumption, however, needs to be done with care.
Making a fairness assumption can lead to patently false results.
This applies for instance to the alternating bit protocol
in cases where one of the possible behaviours of the unreliable
channel is to perpetually lose all messages.

In the study of routing protocols for wireless networks, a desirable
liveness property is \emph{packet delivery} \cite{FGHMPT12a,TR13}.
It says that a data packet injected at a source node will eventually be delivered
at its destination node, provided (a) the source node is connected to the target node
through a sequence of 1-hop connections where each two adjacent nodes
are within transmission range of each other, and (b) no 1-hop connection
in the entire network breaks down before the data packet is delivered.
In a reasonable model of a wireless network all connections between nodes can
nondeterministically appear or disappear in any state. As a consequence,
transitions leading to successful delivery of a packet remain possible
as long as the packet is not delivered, and for this reason the assumption of
full fairness is sufficient to establish packet delivery
even without the side conditions (a) and (b).\footnote{All relevant transitions up to the final delivery can
  be relabelled $\tau$, and this cloud of $\tau$-transitions has only one exit, successful delivery,
  which remains reachable throughout. Full fairness allows abstraction from this cloud of $\tau$-transition.}
Nevertheless, such a result has little bearing on reality in wireless networks.
\newcommand{\defis}{\stackrel{{\it def}}{=}}

I quote an example from \cite{GH18} as another illustration of this phenomenon.\label{AB}
Consider the CCS process $Q:=(X|b)\backslash b$ where \raisebox{0pt}[11.5pt][0pt]{$X \defis a.X + \bar b.X$}.
The syntax and semantics of the process algebra CCS, in which this example is specified,
can be found in Appendix~\ref{CCS}.
{\makeatletter
\let\par\@@par
\par\parshape0
\everypar{}\begin{wrapfigure}[2]{r}{0.385\textwidth}
 \vspace{-4ex}
  \input{phone}
  \centerline{\raisebox{1ex}{\box\graph}}
   \vspace{-10ex}
  \end{wrapfigure}
\noindent
The transition system of this process is depicted on the right.%
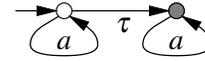
\footnote{States are depicted by circles and transitions by arrows between them. An initial state is
  indicated by a short arrow without a source state. When a liveness property
  modelled as a set of goal states is involved, these states are indicated by shading.}
The process $X$ models the behaviour of relentless Alice, who either picks up her phone 
when Bob is calling ($\bar b$), or performs another activity ($a$), such as eating an apple.
In parallel, $b$ models Bob, constantly trying to call Alice; the action $b$ models the 
call that takes place when Alice answers the phone.
A desired (by Bob) liveness property $\mathcal{G}$ is to achieve
a phone connection with Alice, i.e.\ to reach the rightmost state
in the above transition system.
Under each of the notions of strong and weak fairness reviewed in \cite{GH18},
$Q$ satisfies $\mathcal{G}$.
  Yet, it is perfectly reasonable that the connection is never
  established: Alice could never pick up her phone, as she is not in the mood of talking to Bob; maybe she
  totally broke up with him.
  \par}
My last example is taken from \cite{vG19}.
Suppose Bart stands behind a bar and wants to order a beer. But by lack of any formal queueing protocol
many other customers get their beer before Bart does. This situation can be modelled as a transition
system where in each state in which Bart is not served yet there is an outgoing transition
modelling that Bart gets served, but there are also outgoing transitions modelling that someone else
gets served instead. The essence of fairness is the assumption that Bart will get his beer eventually.
Fairness rules out as unfair any execution in which Bart could have gotten a beer any time, but never will.
Yet, this possibility can not be ruled out in reality.

These examples are not anomalies; they describe a default situation.
A fairness assumption says, in essence, that if you try something often enough, you will eventually
succeed. There is nothing in our understanding of the physical universe that supports such a belief.
Fairness assumptions are justified in exceptional situations, the verification of the alternating
bit protocol being a good example. However, by default they are unwarranted.

\subsection{Probabilistic arguments in favour of fairness}

In defence of making a fairness assumption it is sometimes argued that
whenever at some point the probability of success is 0, the success
possibility should not be part of our model of reality. When
infinitely many choices remain that allow success with a fixed positive
probability, with probability 1 success will be achieved eventually.
This argument rests on assumptions on relative probabilities of
certain choices, but is applied to models that abstract from those probabilities.

My counter-arguments are that (1) when abstracting from probabilities
it is quite well possible that a success probability is always
positive, yet quickly diminishing, so that the cumulative success
probability is less than 1, and (2) that in many applications one does
not know whether certain behaviours have a chance of occurring or
not, but they are included in the model nevertheless.
\vspace{-6pt}

\subsection{Demonic choice in reactive systems}\label{demonic}

\begin{figure}
  \input{a}
  \raisebox{1ex}{\box\graph}
\hfill
  \input{b}
  \raisebox{1ex}{\box\graph}
\hfill
  \input{d}
  \raisebox{1ex}{\box\graph}
  \\[-5ex]
  \input{e}
  \raisebox{1ex}{\box\graph}
\hfill
  \input{c}
  \raisebox{1ex}{\box\graph}
\hfill
  \input{f}
  \raisebox{1ex}{\box\graph}
\vspace{-2pt}
\caption{The meaning of choice in reactive systems}\label{1}
\vspace{-2pt}
\end{figure}
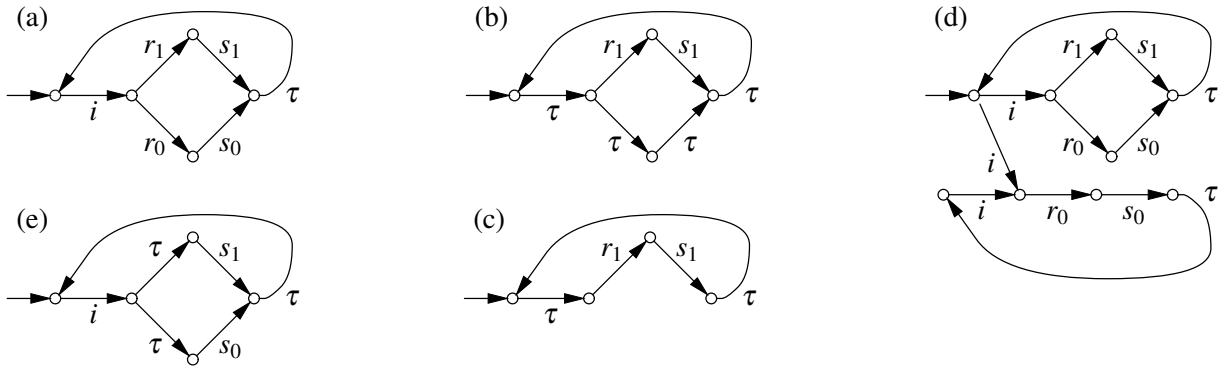

The issue is further illustrated by the labelled transition system of Figure \ref{1}(a).
It depicts a system $B$ that, after an initialisation action $i$, reads a Boolean value
on a certain input channel. The action $r_b$ indicates the reading of $b\in\{0,1\}$.
Subsequently it forwards the received value $b$ on a different channel by performing
the action $s_b$. Then, via an internal action $\tau$, $B$ returns to its initial state.
Thus $B$ functions as a one-bit buffer.

Using \emph{fair abstraction}, e.g., standard process algebraic reasoning as described in
Section~\ref{fair abstraction}, one can prove that this system will eventually receive
and forward the Boolean value 1. Namely, one renames all actions one is not interested in into the
internal action $\tau$, thereby obtaining the system of Figure \ref{1}(b). This system is weakly
(and branching) bisimulation equivalent with the system depicted in Figure \ref{1}(c), and the latter
clearly receives and forwards value 1 eventually.

In my view, this is another example showing that the assumption of fairness is by default
unwarranted. For the reactive system $B$ could well be placed in an environment that will never
provide the value 1 on the modelled input channel.

A counterargument has been brought to my attention by one of the referees of this paper.
Namely a transition system as in Figure \ref{1}(a) denotes a system with the property that each
execution of the action $i$ reaches a state from which it is actually possible that action $r_1$ will occur.
This is at odds with my ``placement'' of $B$ in an environment where this is not going to ever
happen. If one wants to model the possible placement of $B$ in such an environment, one
might need a transition system that looks more like Figure \ref{1}(d). Here one sees the possibility that
during the execution of the action $i$, system $B$ falls into an environment that will never provide the
input 1.

In this point of view, my argument against the use of fairness loses its force.
For the system of Figure \ref{1}(d) surely is not guaranteed to ever receive and forward a 1, regardless
whether one employs fairness of not. The system of Figure \ref{1}(a) on the other hand is one that will not
be placed in such an environment, and here the probabilistic argument could be used to argue that
it is very unlikely that one will never observe $r_1$.

In my opinion\pagebreak[4] the modelling of a reactive system ought to be independent of the environment in which it
is going to operate. Figure \ref{1}(d) shows a bad model of $B$, as it may reach a state where it no
longer accepts the value $1$. The better representation is Figure \ref{1}(a), and its meaning should allow
an environment as discussed above.

The example above involves a choice between $r_0$ and $r_1$ that is entirely triggered by the
environment. This is often called an \emph{external choice} \cite{Ros97}. One may wonder if the same
reasoning applies to an \emph{internal} or \emph{nondeterministic} choice. Figure \ref{1}(e) shows a
system that after the action $i$  nondeterministically chooses to either send $0$ or $1$ on its
output channel. Here one may argue that Figure \ref{1}(e) models a system where after each occurrence of $i$
both choices are open, so that there must be a positive probability that $s_1$ will occur.
My argument (1) above that even this does not guarantee that $s_1$ will ever occur,
amounts to saying that the probability of doing  $s_1$  could diminish quickly after each
occurrence of $i$. My referee answers that in this scenario ``the Markovian assumption is being
violated, meaning that another state, with a different choice semantics, is being entered after each choice''.
This could presumably be modelled by an infinite-state transition system, but not by one such as
Figure \ref{1}(e).

My point of view is that a nondeterministic choice such as depicted in Figure \ref{1}(e) is really an
external choice triggered by an aspect of the environment that is outside our grasp.
Either one has chosen to abstract from this aspect, or one does not know what truly
causes the decision. Thus, Figure \ref{1}(e) can be seen as representing the system from  Figure \ref{1}(a)
but with the actions $r_b$ depicted as $\tau$. Consequently, this system might run in
environments were always the choice leading to $r_0$ is taken. It might also run in environments
where an unknown adversary rolls an increasingly unfair dice at each choice point.
This view of nondeterminism is sometimes called \emph{demonic choice} \cite{MM01}.

\section{Progress and justness}\label{pj}

Having reached the point where I advocate proving liveness properties of distributed system without
resorting to fairness assumption, the question arises whether any weaker assumptions in lieu of
fairness are appropriate. My answer is a resounding \emph{yes}. The least that should always be
assumed when proving liveness properties is what I call \emph{progress}. Without a progress
assumption no meaningful liveness properties can be established.

To illustrate this concept, consider the transition system on the right,
{\makeatletter
\let\par\@@par
\par\parshape0
\everypar{}\begin{wrapfigure}[1]{r}{0.285\textwidth}
 \vspace{-5ex}
 \input{Cataline}\label{Cataline}
  \centerline{\box\graph}
 \end{wrapfigure}
\noindent
taken from \cite{vG19}. It models Cataline eating a croissant in Paris
and abstracts from all activity in the world except the eating of that croissant. It thus has two
states only---the states of the world before and after this event---and one transition.
A possible liveness property $\mathcal{G}$ says that the croissant will be eaten. It corresponds
with reaching state $2$. This liveness property does not hold if the system may remain forever in
the initial state $1$. The assumption of progress rules out that behaviour.
In the context of \emph{closed systems}, having no run-time interactions with the environment,
it is the assumption that a system will never get stuck in a state with outgoing transitions.
It is only when assuming progress that $\mathcal{G}$ holds.
\par}

For \emph{reactive systems}, having run-time interactions with their environment, the progress
assumption as formulated above would rule out too many behaviours. Take for instance the one-bit
buffer of Figure \ref{1}(a). When assuming that the system can not get stuck in a state with outgoing
transitions, it would follow that either $s_0$ or $s_1$ will occur. Yet, in real life the system may be
stuck in the state right after $i$, due to the environment not providing any value on the system's
input channel.  In general, a transition may represent an interaction between the distributed
system being modelled and its environment. In many cases it can occur only if both
the modelled system \emph{and} the environment are ready to engage in it. I therefore distinguish
\emph{blocking} and \emph{non-blocking} transitions \cite{GH15a}.%
\footnote{In \cite{TR13} the \emph{internal} and \emph{output} transitions constitute the
  non-blocking ones. In \cite{Rei13} blocking and non-blocking transitions are called \emph{cold}
  and \emph{hot}, respectively.}  A transition is non-blocking if
the environment cannot or will not block it, so that its execution is entirely under the
control of the system under consideration. A blocking transition on the other hand may fail to occur
because the environment is not ready for it. 
In \cite{GH18}, paraphrasing \cite{TR13,GH15a}, the assumption of progress was formulated as follows:
\begin{quote}\it
A (transition) 
system in a state that admits a non-blocking transition will eventually progress, i.e., perform a transition.
\end{quote}
In other words, a run will never get stuck in a state with outgoing non-blocking transitions.

\phantomsection

Justness is an assumption strengthening progress, proposed in \cite{TR13,GH15a,GH18}.
It can be argued that once one adopts progress it makes sense to go a step further and
adopt even justness.

The transition system on the right models Alice making an unending sequence\label{AC}
{\makeatletter
\let\par\@@par
\par\parshape0
\everypar{}\begin{wrapfigure}[5]{r}{0.187\textwidth}
 \vspace{-4ex}
 \input{Alice}
 \centerline{\box\graph}
 \vspace{3ex}
  
 \input{AliceCataline}
 \centerline{\box\graph}
 \end{wrapfigure}
\noindent
of phone calls in London.
There is no interaction of any kind between Alice and Cataline.
Yet, I may choose to abstract from all activity in the world except the eating of the croissant by
Cataline, and the making of calls by Alice.
This yields the combined transition system on the bottom right.
Even when taking the {\it cr}-transition to be non-blocking,
progress is not a strong enough assumption to ensure that Cataline will ever eat the croissant.
For the infinite run that loops in the first state is progressing.
Nevertheless, as nothing stops Cataline from making progress, in reality {\it cr} will occur. \cite{GH18,vG19}
\par}
This example is not a contrived corner case, but a rather typical illustration of an issue that is
central to the study of distributed systems. Other illustrations of this phenomenon occur in
\cite[Section 9.1]{TR13}, \cite[Section 10]{GH15b}, \cite[Section 1.4]{vG15} and \cite[Section 4]{EPTCS255.2}.
The assumption of justness aims to ensure the liveness property occurring in these examples.
In \cite{GH18} it is formulated as follows:
\begin{quote}\it
  Once a non-blocking transition is enabled that stems from a set of parallel components,
  one (or more) of these components will eventually partake in a transition.
\end{quote}
In the above example, {\it cr} is a non-blocking transition enabled in the initial state.
It stems from the single parallel component Cataline of the distributed system under consideration.
Justness therefore requires that Cataline must partake in a transition. This can only be {\it cr},
as all other transitions involve component Alice only. Hence justness says that {\it cr} must occur.
The infinite run starting in the initial state and not containing {\it cr} is ruled out as
unjust.

A formal definition of justness is supplied in Appendix~\ref{justness}.
I believe reading it is not necessary to understand the forthcoming material.

\section{A hierarchy of completeness criteria}\label{sec:hierarchy}

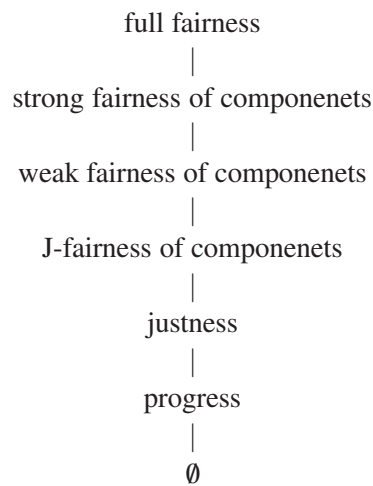
\begin{wrapfigure}[14]{r}{0.48\textwidth}
\input{hierarchy}
\vspace{-6ex}
\centerline{\raisebox{1ex}{\box\graph}}
\vspace{1ex}
\caption{A hierarchy of completeness criteria}
\label{fig:hierarchy}
\end{wrapfigure}
In \cite{vG19}, assumptions like progress, justness and fairness are called \emph{completeness criteria}.
They serve to rule out certain runs of distributed systems that appear to be valid in their
representations as transition systems, on grounds that such runs are assumed not to occur in practice.
The completeness criterion ``progress'' for instance, applied to the transition system on
page~\pageref{Cataline}, where {\it cr} is assumed non-blocking, rules out the run in which
Cataline does not eat her croissant (i.e.\ where no transition is ever taken).

One completeness criterion is called \emph{stronger} than another if it rules out more runs.
The weakest completeness criterion is the empty one ($\emptyset$)---it rules out no runs.\pagebreak[2]
Figure~\ref{fig:hierarchy} orders several completeness criteria on strength.
Here progress and justness are described in Section~\ref{pj} and formalised in Appendix~\ref{justness}.\linebreak[2]
The three forms of fairness of components are formally defined in Appendix~\ref{fairness}, taken from \cite{GH18}.
The hierarchy of Figure~\ref{fig:hierarchy} is supported by Proposition~\ref{hierarchy} in Appendix~\ref{fairness}.
Many other notions of fairness are classified in \cite{GH18}; some of them have a strength incomparable to justness.
One notion of fairness from the literature, named \emph{fairness of events} in \cite{GH18}, is shown
to coincide with justness \cite[Theorem 15.1]{GH18}. However, it has been defined only for the
special case that the set of blocking actions is empty. Fairness of events was first studied in
\cite{CS84}---although on a restriction-free subset of CCS where it coincides with fairness of
components---and then in \cite{CDV06a}, under the name fairness~of~actions.

A liveness property holds for a distributed system iff it holds for each of its runs.
It holds under a progress, justness or fairness assumption, or in general when adopting a
particular completeness criterion, if it holds for all runs that are not ruled out by that
assumption or criterion. So a liveness property is more often satisfied when adopting a stronger
completeness criterion.

The concept of \emph{full fairness} described in Section~\ref{fair abstraction} and illustrated in
Section~\ref{demonic} is also formalised in \cite{GH18}. It is not a fairness notion as defined
in Appendix~\ref{fairness}, since it does not rule out a specific set of runs. Yet its strength can be compared with
other notions of fairness based on which liveness properties are obtained by making this assumption.
Full fairness turns out to be the strongest of all possible fairness assumptions \cite{GH18}.

\section{Process algebras without fairness assumptions, and their limitations}\label{main}

\emph{Strong bisimulation equivalence} \cite{Mi89,vG00}, also called \emph{strong bisimilarity}, is
one of the most prominent semantic equivalences considered in concurrency theory.
Virtually all semantic equivalences or refinement preorders employed in process algebra
are coarser than or equal to strong bisimilarity; that is, they identify strongly bisimilar systems.

Here I will argue that these approaches \emph{cannot}
make sufficiently strong progress assumptions to establish meaningful
liveness properties in realistic applications.
Namely, I will show two programs, $P$ and $Q$, that are strongly bisimilar, and hence
equated in virtually all process-algebraic approaches to date.
Yet, there is a crucial liveness property that holds for $P$ but not for $Q$,
when assuming justness but not fairness. So the process algebra must either claim that both
programs have the liveness property, which in case of $Q$ could be an
unwarranted conclusion, possibly leading to the design of systems with
dangerous or catastrophic behaviour, or it falls short in asserting the
liveness property of $P$.
\\[1ex]
\mbox{}\hfill
$x:= 1\quad \| \quad\textbf{repeat}~~ y:=y+1 ~~\textbf{forever}$
\hfill($P$)
\\[1ex]
Program $P$ is the parallel composition of two non-interacting processes,
one of which sets the variable $x$ to $1$, and the other repeatedly
increments a variable $y$. I assume that both variables $x$ and
$y$ are initialised to $0$.\linebreak[3]
The reader may assume that $x$ is a local variable maintained by one component, and $y$ by the
other, or that $x$ and $y$ reside in central memories available to both components; but in the latter
case $x$ and $y$ reside in completely disconnected central memories, so that an access to the
variable $y$ by the right component in no way interferes with an access to variable $x$ by the left one.

\begin{wrapfigure}[6]{r}{5cm}
\vspace{-4.5ex}
\begin{tabbing}\qquad\=\textbf{repeat}\\
\> \quad \= \textbf{case} \\
\>\> \quad \= \textbf{if} ~\texttt{True}~ \textbf{then}~ y:=y+1 ~\textbf{fi}\=\\
\>\>\> \textbf{if} ~$x=0$~ \textbf{then}~ x:=1 ~\textbf{fi}\\
\>\> \textbf{end} \\
\>\textbf{forever}\>\>\>(Q)
\end{tabbing}
\end{wrapfigure}
In program $Q$ the \emph{case}-statement is interpreted such that if the
conditions of multiple cases hold, a non-deterministic choice is made
which one to execute. The conditional write \textbf{if} ~$x=0$~ \textbf{then}~ x:=1 ~\textbf{fi}
describes an atomic read-modify-write (RMW) operation\footnote{\url{https://en.wikipedia.org/wiki/Read-modify-write}}.
Such operators, supported by modern hardware, read a memory location and simultaneously write a new value into it
that may be a function of the previous value.

The programs $P$ and $Q$ are strongly bisimilar; both can be
represented by means of the following labelled transition system:
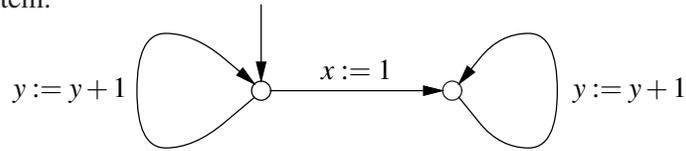
\vspace{-3ex}

\input{PQ}
\centerline{\box\graph}
\vspace{3ex}

As a warm-up exercise, one may ask whether the variable $y$ in $P$ or
$Q$ will ever
reach the value $7$---a liveness property. A priori, I cannot give a
positive answer, for one can imagine that after incrementing $y$ three
times, the program for no apparent reason stops making
progress and does not get around to any further activity. In most
applications, however, it is safe to assume that this scenario will
not occur. To accurately describe the intended behaviour of $P$ or $Q$,
or any other program, one makes a progress assumption as described in Section~\ref{pj},
saying that if a program is in a state where further activity is
possible (and this activity is not contingent on input from the
environment that might fail to occur) some activity will in fact
happen. This assumption is sufficient to ensure that in $P$ or $Q$ the
variable $y$ will at some point reach the value $7$.

Progress assumptions are commonplace in process algebra and many other
formalisms. They are explicitly or implicitly made in CCS, ACP, the
$\pi$-calculus, CSP, etc., whenever such formalisms are employed to
establish liveness properties. Temporal logics, such as LTL
\cite{Pn77} and CTL \cite{EC82}, have progress assumptions built in,
namely by disallowing states without outgoing transitions and evaluating temporal formulas by
quantifying over infinite paths only; they can formalise the statement that $y$ will in fact
reach the value $7$. 

A more interesting question is whether $x$ will ever reach the value $1$.
This liveness property is \emph{not} guaranteed by progress
assumptions as made in any of the standard process algebras or temporal logics.
The problem is that all these formalisms rest on a model of
concurrency where parallel composition is modelled as arbitrary interleaving.
The programs $P$ and $Q$ have computations like
$$\begin{array}{lllllll}
x:=1;& y:=y+1;& y:=y+1;& y:=y+1;& y:=y+1;& y:=y+1;& \dots \\
y:=y+1;& x:=1;& y:=y+1;& y:=y+1;& y:=y+1;& y:=y+1;& \dots \\
y:=y+1;& y:=y+1;& y:=y+1;& y:=y+1;& x:=1;& y:=y+1;& \dots \\
\end{array}$$
where the action $x:=1$ can be scheduled arbitrary far in the sequence
of $y$-incrementations, but also a computation
\begin{equation}\tag{$C^\infty$}
\begin{array}{lllllll}
y:=y+1;& y:=y+1;& y:=y+1;& y:=y+1;& y:=y+1;& y:=y+1;& \dots \\
\end{array}
\end{equation}
in which $x:=1$ never happens, because $y:=y+1$ is always scheduled instead.
For this reason, temporal logic as well as process algebra---when not
making fairness assumptions---say that $x$ is not guaranteed to reach
the value $1$, regardless whether talking about $P$ or $Q$.

\begin{wrapfigure}[10]{r}{6.2cm}
\vspace{-3ex}\hfill
\begin{tabular}[t]{c|cc|cc}
  Liveness goal: & \multicolumn{2}{c|}{$y=7$} & \multicolumn{2}{c}{$x=1$} \\
  \hline
  Program        & $P$ & $Q$ & $P$ & $Q$ \\
  \hline
  \raisebox{0pt}[12pt][0pt]
  {\it full fairness}  & + & + & + & + \\
  {\it justness}       & + & + & + & $-$ \\
  {\it progress}       & + & + & $-$ & $-$ \\
  $\emptyset$          & $-$ & $-$ & $-$ & $-$ \\
\end{tabular}\renewcommand\figurename{~~~Figure}
\caption{\it Liveness properties obtained \mbox{}~~~as a function of assumptions made}
\label{obtained}
\end{wrapfigure}

When assuming that parallel composition is implemented by
means of a scheduler that arbitrarily interleaves actions from both
processes, this conclusion for $P$ appears plausible. However, when $\|$
denotes a true parallel composition, where the program $P$ consists
of two completely independent processes, it appears more reasonable
to adopt the assumption of justness from Section~\ref{pj}, which
guarantees that $x$ will reach the value $1$.
However, whereas justness disqualifies the computation
$C^\infty$ for $P$, it rightly allows it for $Q$.
As the choice between the two cases of the \textbf{case}-statement is
made by forces outside our grasp---see Section~\ref{demonic}---one may
not rule out the possibility that the first case is chosen every round.

A sufficiently strong fairness assumption would (unjustly) eliminate this
computation even for $Q$---see Figure~\ref{obtained}; here I argue for a
theory of concurrency in which such a fairness assumption is not made.

Hence, virtually all existing process-algebraic approaches equate
two programs, of which one has the liveness property that eventually
$x$ will reach the value $1$, and the other does not, at least not
without assuming fairness. So those approaches that do not assume
fairness \cite{DH84,BHR84,Ho85,Wa90,Ros97,GLT09b} lack the power to establish this property for $P$.

\paragraph{Variations on this example}
Readers that prefer the states in the transition systems for $P$ and $Q$ to be valuations of the
variables $x$ and $y$ may easily unfold the given transition system into an infinite-state
one with this property. This unfolding preserves the state of affairs that $P$ and $Q$ are strongly bisimilar
systems of which one has a liveness property that the other lacks (when assuming justness but not fairness).

For readers that dislike the RMW operations in $Q$, one can easily skip the preconditions
\texttt{True} and $x=0$. To reobtain bisimilarity, $P$ needs then be changed into
\\[1ex]
\mbox{}\hfill
$\textbf{repeat}~~x:= 1 ~~\textbf{forever}\quad \| \quad\textbf{repeat}~~ y:=y+1 ~~\textbf{forever}$
\hfill($P'$)
\\[1ex]
The resulting labelled transition system of both $P'$ and $Q'$ features one state and two loop-transitions.
Again one obtains two bisimilar systems of which only one has the liveness property that 
$x$ will be $1$.

I have presented this example in pseudocode to stress that the problem is not specific to process
algebras. However, it can also be expressed in process algebras like CCS\@. Let $Q$ for instance be 
the process $Q:=(X|\bar b)\backslash b$ discussed on page~\pageref{AB}, and let $P=(Y|\tau)$ with
\plat{$Y\defis a.Y$} be the parallel composition of Alice and Cataline discussed on page~\pageref{AC}, but
with {\it cr} rendered as $\tau$. Both systems can be represented by the labelled transition system
drawn on page~\pageref{AB}. Hence they are strongly bisimilar. Yet, when assuming justness but not
fairness, $P$ has the liveness property that the second state will be reached, whereas $Q$ does not.

\section{A research agenda}

This brings me to my research agenda in this matter: the development
of a theory of concurrency that is equipped to ensure liveness
properties of distributed systems, incorporating justness assumptions
as explained above, but without making fairness assumptions.  This
theory should encompass process algebra, temporal logic, Petri nets and
other semantic models, as well as treatments of real-time, and of the
interaction between probabilistic and nondeterministic choice.

Since this involves distinguishing programs that are strongly
bisimilar, it requires a complete overhaul of the basic
machinery that has been built in the last few decennia.
It requires new equivalence relations between processes, new
axiomatisations, new induction principles to reason about infinite
processes, new congruence formats for operational semantics ensuring
compositionality of operators, and matching extensions with time and
probabilities.

As in the absence of fairness assumptions some crucial systems like
fair schedulers or mutual exclusion protocols cannot be accurately
specified in Petri nets or standard process algebras \cite{Vo02,KW97,GH15b}, it
also involves the study of adequate model or language extensions, and
their expressive power.

My agenda furthermore aims at developing a methodology that allows successful
application of the envisioned theory of concurrency to the specification, analysis and
verification of realistic distributed systems, focusing on cases where
the new balance in establishing liveness properties bears fruit.
An example of this is the analysis of routing protocols in wireless networks, where the
packet delivery property from Section~\ref{dangers} can hold in a meaningful way
only when assuming justness but not full fairness.

\newpage

\noindent
This research agenda involves the following tasks:
\vspace{-1pt}\leftmargini 20pt
\begin{enumerate}\parskip 0pt\itemsep 0pt
\item Setting up a framework for modelling specifications and implementations of
  distributed systems that encompasses justness without making global fairness assumptions.
\item To investigate and classify semantic equivalences (necessarily finer than or
  incomparable with strong bisimilarity) that respect liveness
  when assuming justness but not fairness.
\item To study liveness and justness properties in non-interleaved
  semantic models like Petri nets, event structures and higher
  dimensional automata.
\item To find complete axiomatisations and adequate induction
  principles for process algebras with justness.
\item To find syntactic requirements for the operational specification of operators
  that guarantee that relevant justness-preserving equivalences are congruences.
\item To study the necessary extensions to process algebras or Petri nets
  to model simple systems like fair schedulers, and investigate the
  relative expressiveness of process algebras with and without them.
\item To re-evaluate the possibility and impossibility results for
  encoding synchrony in asynchrony when insisting that justness
  properties are preserved.
\item To extend relevant justness-preserving formalisms with treatments
  of real-time.
\item To adapt the existing testing theory for nondeterministic
  probabilistic processes to a setting where justness is preserved.
\item To apply the obtained formalisms to (dis)prove liveness properties for real distributed systems.
\end{enumerate}

\newcounter{wp}
\newcommand{\workp}[1]{\subsection*{Task \thewp\refstepcounter{wp}: #1}}

\noindent
Below I describe these tasks in more detail.

\refstepcounter{wp}
\workp{A framework for modelling distributed systems}

Process algebra remains my favourite framework for modelling specifications and
implementations of distributed systems. When the aim is to establish
liveness properties, the semantics of process-algebraic specification
languages should come with a precise definition of which  runs count as just.
The paper \cite{GH18} provides a general definition of justness, which is
applied to CCS and several of its extensions in \cite{vG19}. It appears that
the same style of definition can easily be applied to process algebras involving
a CSP-style communication mechanism \cite{Ho85,Ros97} or name-binding \cite{Mi99,SW01},
as well as extensions of traditional process algebras with data \cite{GM14}.
Of course this needs to be checked carefully.
For some less usual process algebras it is not yet clear how to formalise the
concept of justness. This applies in particular to process algebras with a priority
mechanism \cite{CLN01}.

\begin{open}
Formally define justness for process algebras with priorities.
\end{open}

As pointed out in Section~\ref{liveness}, for some applications it is warranted
to make a fairness assumption. In \cite{GH18} we make a distinction between
\emph{local} and \emph{global} fairness assumptions. The latter apply globally
to all scheduling problems of a given kind that appear in a distributed system;
these are the kind of fairness assumptions classified in \cite{GH18} and in
Section~\ref{sec:hierarchy} of the present paper. A \emph{local} fairness
assumption, on the other hand, can be justified for a particular
scheduling problem in a system under consideration, and does not apply
automatically to other scheduling problems of the same kind.
In general I think it is warranted to employ local fairness assumptions in
specific circumstances, on top of a global justness assumption.
A good way to formalise this is to use specifications of distributed systems of
the form $(P,\F)$, where $P$ is an expression in a suitable process algebra,
whose semantics might consist of a labelled transition system, equipped with a
classification of some of its runs as just, and $\F$ is a
\emph{fairness specification}, ruling out some of the runs as unfair.
The fairness specification can for instance be given as a collection of formulas
in a temporal logic, each formalising a local fairness assumption.
The runs of the system that count for validating liveness properties then are
those that are both just and not ruled out by $\F$. This mode of specification
is discussed in greater detail in \cite{GH15a}, where also a consistency criterion on
tuples $(P,\F)$ is formulated. It has been used earlier in the specification
language TLA$^+$ \cite{La02}.
It also is the type of specification applied in \cite[Section~9]{TR13}
for the formal specification of the Ad hoc On-demand Distance Vector (AODV)
protocol~\cite{rfc3561}, a wireless mesh network routing protocol,
using the process algebra AWN \cite{FGHMPT12a} for the first part and LTL~\cite{Pn77} for the second.
Interestingly, the local fairness assumptions we needed in that work can be
positioned strictly between strong and weak fairness. I conjecture that this
will be the case for many applications. In \cite[Section~7]{GH18} this form of
fairness was formalised in a general setting, and called \emph{strong weak fairness}.

\workp{A classification of semantic equivalences and preorders}

A crucial prerequisite for verifying that an implementation meets a
specification is a definition of what this means. Such a definition
can be given in the form of an equivalence relation or preorder on
a space of models that can be used to describe both specifications and
implementations. For sequential systems, an overview of suitable
preorders and equivalence relations defined on labelled transition
systems is given in \cite{vG01,vG93}. Preorders and equivalences
specifically tailored to preserve safety and liveness properties are
explored in \cite{vG10}. Equivalences for non-sequential systems are
discussed, e.g., in \cite{GG01}. They include \emph{interleaving equivalences},
in which parallelism is equated with arbitrary interleaving, as well
as equivalence notions that take, to some degree, concurrency
explicitly into account. In \cite{vG15} I show that none of these
equivalences respect \emph{inevitability} \cite{MOP89}: when assuming
justness but not fairness, they all equate systems of which only one
has the property that all its runs reach a specific success state.
Hence, none of these equivalences are suitable for a process-algebraic
framework destined to establish liveness properties under the justness
assumption advocated above.

\begin{open}
  Find and classify suitable semantic equivalences that
  respect liveness when assuming justness but not fairness.
\end{open}

As shown in \cite{vG10}, safety and liveness properties are intimately
linked with the notions of may- and must-testing of De Nicola \&
Hennessy \cite{DH84}. However, \cite{vG10} also treats
\emph{conditional liveness} properties that surpass the power of
must-testing. In \cite{vG19b} I propose a notion of \emph{reward testing},
and show that it matches with conditional liveness properties.
Similar testing frameworks can be applied to derive
preorders for concurrent processes that respect (conditional) liveness
properties in the presence of the justness assumption.
This may yield a result similar to the fair failure preorder of \cite{Vo02}.

Additionally, variants of most existing preorders and equivalences may
be found that respect liveness under justness assumptions.
Strong bisimulation, for instance, induces a relation between the runs of two
bisimilar systems.\footnote{One could simply say that two paths $\pi_1$ and
  $\pi_2$ are related iff for all $n$ the $n^{\rm th}$ state of $\pi_1$ is related to the 
  $n^{\rm th}$ state of $\pi_2$, and the actions between states are the same
  too. An alternative is to relate fewer paths, namely only those that are
  matched by a strategy for the bisimulation game as proposed in
  \cite{HR00}---they did this to capture fairness rather than justness.}
A bisimulation may be called \emph{justness preserving} if it
relates just runs with just runs only. Now justness-preserving strong
bisimilarity is a finer variant of strong bisimilarity that respects liveness when
assuming justness but not fairness.

Possibly, just forcing an existing preorder to respect liveness by
adding appropriate clauses to its definition as sketched above gives a result that is
less suitable for verification tasks. The resulting relation may be hard to
decide, and may fail to satisfy the \emph{recursive specification principle} (RSP) of \cite{BK86a,BBK87a,GV93},
saying that guarded recursive specifications have unique solutions.
This principle plays a central r\^ole in verification by means of equivalence
checking \cite{Ba90,GV93,GM14}. To sketch the problem, consider the processes $P$ and
$Q$ of Section~\ref{main}. In strong bisimulation semantics these processes are
identified, and this can be shown by means of RSP\@.
For consider the system of the following two equations:
\[  U = (y:=y+1).U + (x:=1).V \qquad V = (y:=y+1).V \]
where $(y:=y+1)$ and $(x:=1)$ are simply treated as atomic actions.
This system of equations is \emph{guarded}, as defined in \cite{Mi89,BK86a,BBK87a}.
Moreover, $P$ is a solution for this system of equations up to strong bisimilarity,
in the sense that if one substitutes $P$ for the first process variable $U$ and
something suitable for the variable $V$, one obtains two statements that hold when
interpreting $=$ as strong bisimilarity. Similarly, also $Q$ is a solution.
Based on this, RSP tells that $P$ and $Q$ are strongly bisimilar with each other.
This is one of the standard ways to show the semantic equivalence of
specifications and implementations.
Now when moving from ordinary to justness-preserving strong bisimilarity,
the processes $P$ and $Q$ are no longer equivalent. For a bisimulation would
relate the run of $P$ that forever takes the left $y:=y+1$-loop, with the
corresponding run of $Q$, and thus relates an unjust run to a just one.
Nevertheless both $P$ and $Q$ turn out to be solutions to the above system of
equations even up to justness-preserving strong bisimilarity.
So applying RSP would yield the wrong conclusion that $P$ and $Q$ are equivalent.
It follows that RSP, while sound for strong bisimilarity, is not sound for
justness-preserving strong bisimilarity. For the same reasons it appears to be
unsound for the fair failure preorder alluded to above.
This robs us of a valuable verification tool.

This problem might be addressed by formulating more discriminating
preorders and equivalences that do not feature explicit conditions on
infinite runs, yet respect liveness properties.
An idea might be a version of strong bisimilarity that also takes the component labels
of Appendix~\ref{CCS} into account, rather than merely the action labels.
While such an equivalence surely preserves justness, it may be considered too
fine, in that it distinguishes processes that for all practical purposes should
be identified. Examples are $P|Q \neq Q|P$ and $P|{\bf 0}\neq P$.
A suitable semantic equivalence would be less discriminating than that, but, in
some aspects more discriminating than justness-preserving bisimilarity.

Another equivalence that respects liveness when assuming justness but not fairness
is the \emph{structure preserving bisimilarity} of \cite{vG15}. That equivalence is
most likely also too discriminating for many verification tasks, so more
research is called for. Location based equivalences \cite{BCHK94} might be of use here.

\workp{Petri nets and other semantic models}

The standard semantics of process algebras is in terms of labelled
transition systems. However, for accurately capturing causalities
between event occurrences, models like Petri nets~\cite{Rei13}, event
structures~\cite{Wi87a} or higher dimensional automata~\cite{Pr91a,vG91}
are sometimes preferable. As shown above,
unaugmented labelled transition systems are not sufficient to capture
liveness properties when assuming justness but not fairness.  On the
other hand, Petri nets naturally offer a structural characterisation of
justness: if a transition is enabled, and none of the tokens enabling
it are ever consumed by a competing transition, then it will
eventually fire.

\begin{open}
Is the structural characterisation of justness from Petri nets consistent
with the characterisations of justness for process algebras from \cite{GH15a,GH18,vG19}?
\end{open}
To be precise, a process-algebraic expression $P$ is translated into a Petri
net $\den{P}_{\rm PN}$ through the standard Petri nets semantics of \cite{GM84,Wi84,GV87,Old91}.
The just runs of $\den{P}_{\rm PN}$ are structurally determined. So
$\den{P}_{\rm PN}$ translates further to labelled transition system
$\den{\den{P}_{\rm PN}}_{\rm LTS}$ in which some of the runs are marked as just.
On the other hand, using for instance the component-enriched structural
operational semantics of Appendix~\ref{CCS}, $P$ translates directly into such a
labelled transition system $\den{P}_{\rm LTS}$.
The question now is whether $\den{\den{P}_{\rm PN}}_{\rm LTS}$ is semantically
equivalent to $\den{P}_{\rm LTS}$.
Of course this question can be addressed only when a suitable equivalence has been chosen.

Another interesting question is whether event structures or higher dimensional
automata also offer structural characterisations of justness.

\workp{Complete axiomatisations and induction principles}

Many process-algebraic verifications \cite{Ba90} employ principles
like the \emph{recursive specification principle} (described in Task~2)
and the \emph{approximation induction principle} \cite{BK86a,BBK87a},
allowing to derive properties of infinite systems through analysis of
their finite approximations. As argued above, it is likely that these principles do not
hold in straightforward variants of existing semantic equivalences
that respect liveness when assuming justness but not fairness. The two ways to cope
with that are (1) searching for finer equivalences that do not have
this shortcoming, or (2) searching for alternative induction principles that
hold and are useful in verification.
At this time I cannot say which of these directions is the most
promising; more research is in order here.

Algebraic laws have also shown their use in verification, and the
isolation of a complete collection of such laws is often the starting
point of both a good verification toolset and a better understanding
of the semantic concepts involved. For these reasons, finding complete
axiomatisations of suitable equivalences to deal with justness and
liveness is an important task.

\begin{open}
  Find complete axiomatisations and useful induction principles
  for suitable equivalences that respect liveness when assuming justness but not fairness.
\end{open}

\workp{Congruence formats for structural operational semantics}

In process-algebraic verification it is essential that composition
operators on processes, such as the parallel composition, are
\emph{compositional} w.r.t.\ the semantic equivalence employed.
This means that the composition of two processes, each given as an
equivalence class of, say, labelled transition systems, is independent
on the choice of representatives within these equivalence classes.
Compositionality of an operator w.r.t.\ an equivalence is the same as
the equivalence being a congruence w.r.t.\ the operator.

Starting with \cite{dS85,BIM95,GrV92-lb}, the most elegant and efficient
way to establish compositionality results in process algebra is by
means of \emph{congruence formats}, sets of syntactic restrictions
on the operational specification of the behaviour of composition
operators (i.e.\ on rules like the ones of Table~\ref{tab:CCS}) that ensure
compositionality. This line of research is continued in
\cite{Gr93,BolG96,Ul92,Ver95,Bl95,Ul00,UP02,UY00,Fok00b,BFG04,vG11,GMR06,FGW12,vG17b}.

\begin{open}
Find congruence formats tailored to the equivalences produced by Task 2.
\end{open}

\workp{Expressiveness}

In the absence of fairness assumptions even simple systems like fair
schedulers or mutual exclusion protocols cannot be accurately
specified in standard process algebras like CCS or in Petri nets.
This is shown in \cite{Vo02,KW97,GH15b}.
However, these systems \emph{can} be accurately specified in process algebras that feature
\emph{non-blocking reading} \cite{CDV09,EPTCS54.4}, i.e.\ where one can model write actions to a
shared memory that can not be blocked/delayed by read actions to that memory. Such process algebras
include an extension of PAFAS \cite{CDV09}, as well as extensions of CCS with
broadcast communication \cite{GH15a}, priorities \cite{GH15b} or signals \cite{EPTCS255.2}.

When assuming fairness, fair schedulers or mutual exclusion protocols can be correctly specified in
CCS \cite{CorradiniEtAl09}.
But as argued in Section~\ref{dangers}, a fairness assumption is not warranted here.
Fair schedulers, in particular, are used to implement fairness assumptions made in specifications of systems;
assuming fairness to show that a fair scheduler operates as intended totally defeats this purpose.

When assuming only progress and not justness, there is no hope of ever specifying a correct fair
scheduler or mutual exclusion protocol, not even in the mentioned extensions of CCS\@.
It is for this reason that this problem falls within the scope of the present paper.

Since there are at least three extensions of CCS that are expressive enough to model
fair schedulers and mutual exclusion protocols, while CCS itself lacks the required expressiveness,
the relative expressiveness of these extensions is an interesting question.
It is an issue that needs to be studied along with other arguments for one specification formalism over another.

\begin{trivlist} \item[\hspace{\labelsep}\bf Open problem 6a]
Compare the relative expressiveness of the process algebras that can capture mutual exclusion.
\end{trivlist}
\refstepcounter{open}

The resulting study on the relative expressiveness of process
algebras ought to be placed within the formal frameworks developed for
comparisons of expressive power provided in \cite{Bo85,Gorla10a,vG12}.
That is, to show that a specification formalism B is at least as expressive as a 
specification formalism A, a \emph{valid encoding} from A into B needs to be presented.
To show that B has not all the expressiveness of A one shows that no such encoding exists.
Here a valid encoding from A into B is defined as a translation that satisfies some properties.
In the work of Gorla \cite{Gorla10a} this amounts to five correctness criteria on the translation.
In \cite{Bo85,vG12} on the other hand, it is required that for any specification $P$ of a system in
formalism A, the translation of $P$ into formalism B is semantically equivalent to $P$.
This form of validity is parametrised by the choice of a semantic equivalence that spans the
semantic domains in which the languages A and B are interpreted.

When comparing languages like CCS and its extensions mentioned above, preservation of properties
like justness is essential. So for such purposes, another criterion needs to be added to the
definition of a valid translation in the sense of \cite{Gorla10a}, namely a criterion that
guarantees that a system $P$ and its translation have the same liveness properties when assuming
justness and not fairness. When working with valid encodings as in \cite{Bo85,vG12}, the semantic
equivalence that is chosen as parameter in the comparison should be one as found under Task 2.

In fact, when sharpening the concept of a valid encoding in this way, many relative expressiveness
results established in the literature need to be reconfirmed or sharpened as well.

\begin{trivlist} \item[\hspace{\labelsep}\bf Open problem 6b]
Recast relative expressiveness results from the literature in a justness-respecting framework.
\end{trivlist}
In \cite{dS85,vG94a} for instance, results are obtained saying that all languages
with a structural operational semantics of a certain form can be
translated into versions of the process algebras {\sc Meije} \cite{AB84}
and ACP \cite{BK86a}, respectively. These translations are correct up to strong bisimilarity.
An interesting question is what happens to such results when strengthening the strong bisimilarity
in a justness-preserving way. Additionally, one may wonder if some process algebra upgraded with
signals, broadcast communication or priorities may play a similar unifying r\^ole in a
justness-preserving setting.

\begin{trivlist} \item[\hspace{\labelsep}\bf Open problem 6c]
Find a simple process algebra in the style of {\sc Meije} such that all common process algebras can
be translated in to it, where the translation preserves strong bisimilarity as well as liveness
properties when assuming justness but not fairness.
\end{trivlist}

The expressiveness of models of concurrency like Petri nets, event structures and higher
dimensional automata, is also an interesting problem. The extensions of CCS with 
broadcast communication, priorities or signals resist translation into the default incarnations of
these models. CCS with signals appears to be expressible into Petri nets extended with read arcs
however \cite{Vo02}. 

\begin{trivlist} \item[\hspace{\labelsep}\bf Open problem 6d]
Find extensions of the standard models of Petri nets, event structures and higher
dimensional automata that capture broadcast communication, priorities and/or signals.
If possible, show that the resulting models are ``fully expressive'' in some sense.
\end{trivlist}

\workp{Asynchronous interaction in distributed systems}

In \cite{GGS13}, a precise characterisation is given of those distributed systems,
modelled as structural conflict nets, a large class of Petri nets,
that can be implemented without using synchronous communication.
This is part of research aiming to determine to what extent synchronous
communication can be simulated by asynchronous communication.
In this work an original net and its asynchronous implementation are compared
by means of a semantic equivalence that takes divergence, branching time and
causality to some extent into account. It turned out that the result was to a
large extent independent on the precise choice such an equivalence, i.e.\ on the
degree to which it takes divergence, branching time and causality into account.
This result needs to be revisited when using semantic
equivalences with the appropriate respect for justness.
It would be interesting to see if this changes the class of distributed systems
that have asynchronous implementations.

\begin{open}
Characterise the class of structural conflict nets that have asynchronous
implementations under an equivalence that not only takes divergence, branching
time and causality to some extent into account, but also respects liveness when
assuming justness but not fairness.
\end{open}

\workp{Real-time}

When using a process algebra that takes time explicitly into account, justness
may be obtained as a derived concept: a run is just iff time grows unboundedly.%
\footnote{A link between fair runs and runs where time grows unboundedly was made in \cite{Ly96}.}
As a consequence, extra structure to model justness may not be needed.
Naturally, one may wonder if a given untimed semantics of a process algebra featuring
justness can be obtained through an extension of the model with time;
and vice versa how a notion of a just run obtained from a given timed process
algebra can be characterised when abstracting from time.
A version of the first question has already been addressed in \cite{CDV06a}
in the context of the process algebra PAFAS\@.

\begin{open}
Establish the relationship between notions of justness defined explicitly in
untimed process algebra, and those derived from everlasting runs in timed extensions.
\end{open}

\workp{Extensions with probabilistic choice}

Many semantic equivalences for distributed systems have been extended to a
setting featuring both nondeterminism and probabilistic choice.
Prominent examples are the probabilistic bisimulation equivalences of \cite{Seg96}, and the may- and
must-testing equivalences defined in \cite{WL92} and characterised in \cite{DGHM08}.
The latter work determines the coarsest semantic equivalences for
probabilistic processes that respect safety and liveness properties, respectively, when
assuming progress but not justness.
A natural task is to redo this work when assuming justness.

\begin{open}
Characterise the coarsest semantic equivalence for nondeterministic
probabilistic processes that respect liveness properties when
assuming justness but not fairness.
\end{open}
More in general, extended the relevant semantic equivalences from Task 2 to a
probabilistic setting.

\workp{Applications}

Having argued that the work outlined above is a necessary step towards formalising liveness
properties for realistic distributed systems, naturally one would like to see applications
of such liveness properties, including proofs that they do or do not hold,
for cases where fairness assumptions are truly unwarranted, and merely assuming progress is insufficient.
One such application concerns the packet delivery property for routing protocols in wireless networks.
I have indicated in Section~\ref{dangers} that assuming fairness allows one to establish versions of
this property that do not hold in reality. On the other hand, without assuming justness, no useful
packet delivery property will ever hold. Assume a scenario where four network nodes are active,
with nodes 1 and 2 within transition range of each other, and outside transmission range of nodes 3
and 4. The packet delivery property says that an attempt of node 1 to deliver a message to node 2
ought to succeed. Yet, there exists an infinite run in which the message of 1 to 2 is never sent,
let alone received, because all that occurs is an infinite sequence of chatter between nodes 3 and 4.
In fact, attempts to properly formalise packet delivery \cite{TR13} were one of the reasons to
formalise the notion of justness in \cite{GH15a,GH18}.

\begin{open}
  Establishing packet delivery for suitable routing protocols for wireless networks.
\end{open}

\section{Conclusion}

I have presented a research agenda aiming at laying the foundations of
a theory of concurrency that is equipped to ensure liveness properties
of distributed systems without making fairness assumptions.  The
agenda also includes the application of this theory to the specification,
analysis and verification of realistic distributed systems.
I have divided this agenda into 10 tasks, each of which involves solving an open problem.
It is my hope that this document stimulates its readership to address some of these
tasks and problems.
\appendix

\section{CCS}\label{CCS}

\newcommand{\ca}{a}
\noindent
CCS \cite{Mi89} is parametrised with sets ${\K}$ of \emph{agent identifiers} and $\A$ of \emph{names};
each $X\in\K$ comes with a defining equation \plat{$X \stackrel{{\it def}}{=} P$} with $P$ being a CCS expression as defined below.
\href{https://en.wikipedia.org/wiki/List_of_mathematical_symbols_by_subject#Set_operations}{\raisebox{0pt}[10pt]{$Act := \A\dcup\bar\A\dcup \{\tau\}$}}
is the set of {\em actions}, where $\tau$ is a special \emph{internal action}
and $\bar{\A} := \{ \bar{\ca} \mid \ca \in \A\}$ is the set of \emph{co-names}.
Complementation is extended to $\bar\A$ by setting $\bar{\bar{\mbox{$\ca$}}}=\ca$.
Below, $\ca$ ranges over $\A\cup\bar\A$, $\alpha$ over $Act$, and $X,Y$ over $\K$.
A \emph{relabelling} is a function $f\!:\A\mathbin\rightarrow \A$; it extends to $Act$ by
$f(\bar{\ca})\mathbin=\overline{f(\ca)}$ and $f(\tau):=\tau$.
The set $\cT_{\rm CCS}$ of CCS expressions or \emph{processes} is the smallest set including:
\begin{center}
\begin{tabular}{@{}lll@{}}
$\sum_{i\in I}\alpha_i.P_i$ & for $I$ an index set, $\alpha_i\mathbin\in Act$ and $P_i\mathbin\in\cT_{\rm CCS}$ & \emph{guarded choice} \\
$P|Q$ & for $P,Q\mathbin\in\cT_{\rm CCS}$ & \emph{parallel composition}\\
$P\backslash L$ & for $L\subseteq\A$ and $P\mathbin\in\cT_{\rm CCS}$ & \emph{restriction} \\
$P[f]$ & for $f$ a relabelling and $P\mathbin\in\cT_{\rm CCS}$ & \emph{relabelling} \\
$X$ & for $X\in\K$ & \emph{agent identifier}\\
\end{tabular}
\end{center}
The process $\sum_{i\in \{1,2\}}\alpha_i.P_i$ is often written as $\alpha_1.P_1 + \alpha_2.P_2$,
$\sum_{i\in \{1\}}\alpha_i.P_i$ as $\alpha_1.P_1$, and $\sum_{i\in \emptyset}\alpha_i.P_i$ as ${\bf 0}$.
Moreover, one abbreviates $\alpha.{\bf 0}$ by $\alpha$, and $P\backslash\{\ca\}$ by $P\backslash \ca$.
The semantics of CCS is given by the labelled transition relation
$\mathord\rightarrow \subseteq \cT_{\rm CCS}\times Act \color{red} \times \Pow(\Ce) \color{black} \times\cT_{\rm CCS}$, where transitions 
\plat{$P\goto[C]{\alpha}Q$} are derived from the rules of \autoref{tab:CCS}.
\begin{table*}[t]
\caption{Structural operational semantics of CCS}
\label{tab:CCS}
\normalsize
\begin{center}
  \newcommand\mylabel\weg
  \renewcommand\eta\alpha
  \renewcommand\ell\alpha
\framebox{$\begin{array}{c@{\quad}c@{\qquad}c}
&\sum_{i\in I}\alpha_i.P_i \goto[\{\varepsilon\}]{\alpha_j} P_j\makebox[20pt][l]{~~~~($j\in I$)} \\[4ex]
\displaystyle\frac{P\goto[C]{\eta} P'}{P|Q \goto[\Left\cdot C]{\eta} P'|Q} \mylabel{Par-l}&
\displaystyle\frac{P\goto[C]{\ca} P' ,~ Q \goto[D]{\bar{\ca}} Q'}{P|Q \goto[\Left\cdot C \;\cup\; \R\cdot D]{\tau} P'| Q'} \mylabel{Comm}&
\displaystyle\frac{Q\goto[D]{\eta} Q'}{P|Q \goto[\R\cdot D]{\eta} P|Q'} \mylabel{Par-r}\\[4ex]
\displaystyle\frac{P \goto[C]{\ell} P'}{P\backslash L \goto[C]{\ell}P'\backslash L}~~(\ell,\bar{\ell}\not\in L) \mylabel{Res}&
\displaystyle\frac{P \goto[C]{\ell} P'}{P[f] \goto[C]{f(\ell)} P'[f]} \mylabel{Rel} &
\displaystyle\frac{P \goto[C]{\alpha} P'}{X\goto[C]{\alpha}P'}~~(X \stackrel{{\it def}}{=} P) \mylabel{Rec}
\end{array}$}
\vspace{-2ex}
\end{center}
\end{table*}
The second labels {\color{red}$C\in \Pow(\Ce)$}, displayed in red, are not part of these
transitions; they will be introduced, with the definition of $\Ce$, below.
Such a transition indicates that process $P\in \cT_{\rm CCS}$ can perform the action $\alpha\in Act$
and thereby transform into process $Q\in \cT_{\rm CCS}$.
The process $\sum_{i\in I}\alpha_i.P_i$ performs one of the actions $\alpha_j$ for $j\in I$ and subsequently acts as $P_j$.
The parallel composition $P|Q$ executes an action from $P$, an action from $Q$, or in the case where
$P$ and $Q$ can perform complementary actions $c$ and $\bar{c}$, the process can perform a
synchronisation, resulting in an internal action $\tau$. 
The restriction operator $P \backslash L$ inhibits execution of the actions from $L$ and their complements. 
The relabelling $P[f]$ acts like process $P$ with all labels $\alpha$ replaced by $f(\alpha)$.
Finally, the rule for agent identifiers says that an agent $X$ has the same transitions as the body $P$ of its defining equation.
The standard version of CCS \cite{Mi89} features a \emph{choice} operator $\sum_{i\in I}P_i$; here
I use the fragment of CCS that merely features guarded choice.

\paragraph{Components}
The second label of a transition indicates the set of (parallel) \emph{components} involved in
executing this transition. The set $\Ce$ of components is defined as $\{\Left,\R\}^*$, that is, set
of strings over the indicators $\Left$eft and $\R$ight, with $\varepsilon\in\Ce$ denoting the empty sequence
and $\textsc{d}\cdot C := \{\textsc{d}\sigma \mid \sigma\in C\}$ for $\textsc{d}\in\{\Left,\R\}$.
The process $Q$ from page~\pageref{AB} for instance has the transitions \plat{$Q \goto[\{\Left\}]{a} Q$}
and \plat{$Q \goto[\{\Left,\R\}]{\tau} (X|{\bf 0})\backslash c$}. The first transition denotes Alice
performing $a$ny other activity; it involves the left component of the parallel composition only. The second transition
denotes a call between Alice and Bob; it involves both components.
The idea to extend the structural operational semantics of CCS with a component labelling as
indicated in \autoref{tab:CCS}, to achieve an elegant formalisation of justness, stems from Victor Dyseryn [personal communication].

\section{Justness}\label{justness}

This appendix interprets each CCS process as a state in a component-labelled transition system.
This is an ordinary labelled transition system, upgraded with a labelling of the transitions $t$ by the
parallel components involved in performing $t$. It also enriches such systems with the set $B$
of blocking actions. Subsequently, it presents the definitions of a path and the completeness
criteria progress and justness, thereby formalising the intuitions from Section~\ref{pj}.

\begin{definition}{LTS}
A \emph{component-labelled transition system} (CLTS) is a tuple $(S, \Tr, \source,\target,\ell,B,\linebreak[1]\comp)$
with $S$ and $\Tr$ sets (of \emph{states} and \emph{transitions}), $\source,\target:\Tr\rightarrow S$,
$\ell:\Tr\rightarrow Act$ for a set of actions $Act$,
$B\subseteq Act$ a set of \emph{blocking} actions, and
$\comp:\Tr\rightarrow \Pow(\Ce)\setminus\emptyset$ for some set of components $\Ce$,
such that:
\begin{equation}\label{noninterference}\begin{minipage}{5.2in}{
   for all $t,v\in\Tr$ with $\source(t)=\source(v)$ and $\comp(t)\cap\comp(v)=\emptyset$,\\
   there is a $u\in \Tr$ with $\source(u)=\target(v)$, $\ell(u)=\ell(t)$ and $\comp(u)=\comp(t)$.}
\end{minipage}\end{equation}
\end{definition}
The underlying idea \cite{GH18} is that if a transition $v$ occurs that does not affect components in $\comp(t)$,
then the internal state of those components is unchanged, so any synchronisation between these
components that was possible before $v$ occurred, is still possible afterwards.

Let $\Tr_{\neg B} := \{t\in \Tr \mid \ell(t)\notin B\}$ be the set of \emph{non-blocking} transitions.

The following CLTS serves as a semantic model for CCS:
Take $S$ to be the set $\cT_{\rm CCS}$ of CCS processes,
and $\Tr$ the set of transitions \plat{$P\goto[C]{\alpha}Q$} derivable from the rules of \autoref{tab:CCS}.
For each such $t\in\Tr$ one takes $\source(t):=P$, $\target(t):=Q$, $\ell(t)=\alpha$ and $\comp(t):=C$.
It is straightforward to check that property (\ref{noninterference}) is satisfied (or see \cite{vG19}).
The set $B\subseteq Act$ can be chosen at will, depending on the intended application,
as long as one promises to use only restrictions $P\backslash L$ with $L\subseteq B$ and no
renamings $f$ that rename a non-blocking action into a blocking one. So the default choice for $B$
is $Act\setminus\{\tau\}$.

A CLTS could also have been defined as a triple $(S,\Tr,B)$, with $S$ a set,
$\Tr \subseteq S \times Act \times \Pow(\Ce) \times S$ for sets $Act$ of actions and $\Ce$ of
components, and $B\subseteq Act$.

\begin{definition}{path}
A \emph{path} in a CLTS $(S, \Tr, \source,\target,\ell,B,\comp)$ is an alternating sequence
$s_0\,t_1\,s_1\,t_2\,s_2\cdots$ of states and transitions, starting with a state and
either being infinite or ending with a state, such that $\source(t_i)=s_{i-1}$ and
$\target(t_i)=s_i$ for all relevant $i$.
\end{definition}
A \emph{completeness criterion} is a unary predicate on the paths in a (component-labelled) transition system.

\begin{definition}{progress}
A path in a CLTS is \emph{progressing} if either it is infinite or its last state is
the source of no non-blocking transition $t \in \Tr_{\neg B}$.
\end{definition}
Progress is a completeness criterion.

\begin{definition}{concurrency}
Two transitions $t,u\in\Tr$ are \emph{concurrent}, notation $t \smile u$, if $\comp(t)\cap \comp(u)=\emptyset$,
i.e., they have no components in common.
\end{definition}
If $t$ and $u$ are not concurrent, $t\nconc u$, then $u$ is said \emph{to interfere with} $t$.

\begin{definition}{justness}
  A path $\pi$ in an CLTS is \emph{just}
  if for each transition $t \in \Tr_{\neg B}$ with $s:=\source(t)\in\pi$, a transition $u$ occurs
  in $\pi$ past the occurrence of $s$, such that $t \nconc u$.
\end{definition}
Informally, justness requires that once a non-blocking transition $t$ is enabled,
sooner or later a transition $u$ will occur that interferes with it, possibly $t$ itself.

Note that justness is a completeness criterion stronger than progress.

\paragraph{Historical notes and disclaimers}

\df{justness} of justness stems from \cite{GH18}. There it was formulated for transition systems without
the labelling function $\ell$, and with an initial state $I\in S$. This makes no difference.
Appendix A of \cite{GH18} defines the function $\comp:\Tr\rightarrow\Pow(\Ce)\setminus\emptyset$ for a somewhat different fragment of CCS than
is used here. Although that definition has a rather different style than the one of this paper,
on the intersection of both fragments of CCS the resulting functions $\comp$ are easily seen to be the same.\pagebreak[3]

The original definition of justness applied to CCS stems from \cite{GH15a}. That (coinductive)
definition is in a very different style and does not use the concepts $\comp$ and $\smile$.
In \cite{vG19} it is shown that the concept of justness from \cite{GH15a} is the same as that
from \cite{GH18} and the present paper. Moreover, \cite{vG19} contemplates 5 different concurrency
relations $\smile$ between transitions for full CCS, and shows that through \df{justness} all of
them give rise to the same concept of justness. The concurrency relation of \df{concurrency} is
$\smile'_c$ in \cite{vG19}.

The components $\comp(t)$ of a CCS transition $t$ are necessary participants in the execution of
$t$, in the sense that all of them should be ready in order for the transition to be enabled.
They also are the components that are affected by the execution of $t$, in the sense that $t$ may
induce a state-change within all these components. This situation is common for many process algebras.
In \cite{vG19} however, extensions of CCS are reviewed in which only some components are
necessary and only some are affected. Here, subsets $\npc(t),\afc(t)\subseteq\comp(t)$ of
\emph{necessary participants} and \emph{affected components} are defined, and \df{concurrency}
declares $t\aconc u$ iff $\npc(t)\cap \afc(u)=\emptyset$. In this setting it can be that $u$
interferes with $t$ (notation $t\naconc u$) even though $t$ does not interfere with $u$, thus giving
rise to an asymmetric concurrency relation. The treatment above just deals with the special case that
$\npc(t) = \afc(u) = \comp(t)$ for all $t\mathbin\in\Tr$.

In \cite{vG19} one encounters transition systems featuring \emph{signal transitions} $t$.
Such a transition satisfies $\source(t)=\target(t)$ and does not model an action occurrence or state
change in the represented distributed system. To properly apply the
Definitions~\ref{df:path}--\ref{df:J-fair} to such transition systems, one should first normalise the
transition system by deleting all signal transitions \cite{vG19}.

In \cite{vG19} the component $B$ is absent from the definition of a transition system, meaning that
actions are not a priory distinguished into blocking and non-blocking ones. Instead, the concepts of
progress, justness and fairness are indexed with a $B$: a path is called $B$-just if it satisfies
the definition of justness when taking $B$ to be the set of blocking actions. This makes justness
into a family of predicates on paths, rather than a single predicate.

\section{Fairness}\label{fairness}

To formalise weak and strong fairness I use labelled transition systems $(S,\Tr,\source,\target,\ell,B,\Tk)$
that are augmented with a set $\Tk\subseteq\Pow({\Tr})$ of \emph{tasks}
$T\subseteq {\Tr}$, each being a set of transitions.
Here $S,\Tr,\source,\target,\ell,B$ are exactly as in \df{LTS} of a CLTS.

\begin{definitionA}{\cite{GH18}}{fair}
For such a $\IT=(S,\Tr,\source,\target,\ell,B,\Tk)$, a task $T\mathbin\in\Tk$ is
\emph{enabled} in a state $s\in S$ if there exists a non-blocking transition $t\in T$ with
$\ell(t)\notin B$ and $\source(t)=s$.
The task is said to be \emph{perpetually enabled} on a path $\pi$ in $\IT$, if it is enabled in
every state of $\pi$.
It is \emph{relentlessly enabled} on $\pi$, if each suffix of $\pi$ contains a state
in which it is enabled.\footnote{This is the case if the task is enabled in infinitely many states
of $\pi$, in a state that occurs infinitely often in $\pi$, or in the last state of a finite $\pi$.}
It \emph{occurs} in $\pi$ if $\pi$ contains a transition $t\mathbin\in T\!$.

A path $\pi$ in $\IT$ is \emph{weakly fair} if, for every suffix $\pi'$ of $\pi$,
each task that is perpetually enabled on $\pi'$, occurs in $\pi'$.
A path $\pi$ in $\IT$ is \emph{strongly fair} if, for every suffix $\pi'$ of $\pi$,
each task that is relentlessly enabled on $\pi'$, occurs~in~$\pi'$.
\end{definitionA}
In \cite{GH18} many notions of fairness occurring in the literature were casts as instances of this
definition. For each of them the set of tasks $\Tk$ was derived, in different ways, from some other
structure present in the model of distributed systems from the literature. In fact, \cite{GH18}
considers 7 ways to construct the collection $\Tk\!$, and speaks of fairness of \emph{actions},
\emph{transitions}, \emph{instructions}, \emph{synchronisations}, \emph{components},
\emph{groups of components} and \emph{events}. This yields 14 notions of fairness. To compare them,
each is defined formally on a fragment of CCS, and the 14 fairness notions, together with progress,
justness, and some other completeness criteria, are ordered by strength by placing them in a lattice.

For transition systems $(S,\Tr,\source,\target,\ell,B,\smile,\Tk)$ augmented with a concurrency
relation as well as a set of tasks, the following notion of \J-fairness is proposed in \cite{GH18}:

\begin{definitionA}{\cite{GH18}}{J-fair}
For $\IT=(S,\Tr,\source,\target,\ell,B,\smile,\Tk)$, a task $T\mathbin\in\Tk$ is
enabled \emph{during the execution} of a transition $u\in \Tr$ if there exists a $t\in T$
with $\ell(t)\notin B$, $\source(t)=\source(u)$ and $t \smile u$.
It is \emph{continuously enabled} on a path $\pi$ iff it is enabled in every state
and during every transition of $\pi$.
A path $\pi$ is \emph{\J-fair} if, for every suffix $\pi'$ of $\pi$,
each task that is continuously enabled on $\pi'$, occurs in $\pi'\!$.
\end{definitionA}
The name \emph{\J-fairness} is inspired by the notion of \emph{justice} from 
Lehmann, Pnueli \& Stavi \cite{LPS81}. They called a computation \emph{just}
``if it is finite or if every transition%
\footnote{The notion of ``transition'' from \cite{LPS81} is the same as what I call ``component''.}
which is continuously enabled beyond a certain point is taken infinitely many times.''
What this means exactly depends on how one formalises the notion of a component
being ``continuously enabled''. The literature following \cite{LPS81} has
systematically interpreted this as meaning ``in every state'' (``beyond a certain point''),
thereby translating it as ``perpetually enabled'' from \df{fair}.
This is consistent with the words of \cite{LPS81}, as at some point they write
``[This] is an unjust computation, since [the transition or component] $f_1$ is enabled on all
states in it but is never taken.'' 
This makes justice a form of weak fairness as in \df{fair}.
However, a stricter interpretation of ``continuously enabled'' is given in \df{J-fair}, and this
gives rise to a yet weaker notion of fairness that so far has not received much explicit attention in the literature.

I now instantiate the definitions above for a particular choice of $\Tk$, namely fairness of
components \cite{GH18}. Its definition applies to component-labelled transition systems.
Under fairness of components each component determines a task. A transition
belongs to that task if that component contributes to it.
So $\Tk:= \{T_\sigma \mid \sigma\in\Ce\}$ with $T_\sigma:=\{t\in\Tr \mid \sigma\in\comp(t)\}$.
This is the type of fairness studied in \cite{CS87,CDV06c}; it also
appears in \cite{KdR83,AFK88} under the name process fairness.
Fairness of components can also be regarded as the type of fairness studied in \cite{LPS81},
although that paper does not address synchronisation, and thus deals with the special case that 
$\comp(t)$ is a singleton for all $t\in\Tr$.

\begin{proposition}\rm\label{hierarchy}
On any CLTS, a strongly fair path is always weakly fair, a weakly fair path is always \J-fair,
a \J-fair path---under fairness of components---is always just, and a just path is always progressing.
\end{proposition}
\begin{proof}
W.l.o.g.\ let $\pi$ be a \J-fair path under fairness of components, such that a transition $t$ is
enabled in its first state.\footnote{The general case that $s=\source(t)$ occurs in $\pi$ can be reduced to
this special case by taking the suffix of $\pi$ starting at $s$.}
Assume, towards a contradiction, that $\pi$ contains no transition $w$
with $\comp(t) \cap \comp(w)\neq\emptyset$. Let $\sigma\in\comp(t)$. 
Then, using (\ref{noninterference}), for each transition $v$ occurring in $\pi$,
a transition $u_v$ with $\comp(u_v)=\comp(t)\ni\sigma$ is enabled right after $v$.
Note that $u_v\in T_\sigma$. So the task $T_\sigma$ is enabled in each state of $\pi$.
Since $u_v \smile w$ for each transition $w$ in $\pi$,
the task $T_\sigma$ is also enabled during each transition of $\pi$.
So, by $\J$-fairness, the task $T_\sigma$ must occur in $\pi$, contradicting the assumption.
It follows that $\pi$ is just.

The other three statements of Proposition~\ref{hierarchy} follow immediately from the definitions.
\end{proof}
In fact, all implications of Proposition~\ref{hierarchy} are strict.
Counterexamples against their reverses appear in \cite[Examples 3, 20, 12 and 21]{GH18}.%
\footnote{In \cite{GH15a} is it shown that the concept of justice from \cite{LPS81},
  when interpreting ``continuously enabled'' as in \df{J-fair}, and when translating the notion of
  ``transition'' from \cite{LPS81} by ``abstract transition'' for CCS, as defined in \cite{GH15a},
  coincides with justness. Here an ``abstract transition'' is an equivalence class of
  transitions where two transitions are equivalent if they merely differ in the internal state of a
  component not involved in these transitions. The two $a$-labelled transitions in the CCS process
  $a|b$ for instance are deemed equivalent, because they differ only in the state of the $b$-component.

  This result appears at odds with the strictness of the inclusions of Proposition~\ref{hierarchy}.
  But it is not, for Proposition~\ref{hierarchy} applies to a translation of ``transition'' from
  \cite{LPS81} by ``component'' in CCS, which upon closer inspection of \cite{LPS81} is more accurate.
  A component can be seen as an even larger equivalence class of transitions than an abstract transition.}

\paragraph{Acknowledgement}

This paper benefited greatly from the insightful comments of three referees.

\bibliographystyle{eptcsalphaini}
\newcommand{\etalchar}[1]{$^{#1}$}

\end{document}

%% file: phone.tex
\expandafter\ifx\csname graph\endcsname\relax
   \csname newbox\expandafter\endcsname\csname graph\endcsname
\fi
\ifx\graphtemp\undefined
  \csname newdimen\endcsname\graphtemp
\fi
\expandafter\setbox\csname graph\endcsname
 =\vtop{\vskip 0pt\hbox{%
    \special{pn 8}%
    \special{ar 253 41 41 41 0 6.28319}%
    \special{sh 1.000}%
    \special{pn 1}%
    \special{pa 112 16}%
    \special{pa 212 41}%
    \special{pa 112 66}%
    \special{pa 112 16}%
    \special{fp}%
    \special{pn 8}%
    \special{pa 0 41}%
    \special{pa 112 41}%
    \special{fp}%
    \special{sh 1.000}%
    \special{pn 1}%
    \special{pa 355 143}%
    \special{pa 282 70}%
    \special{pa 381 100}%
    \special{pa 355 143}%
    \special{fp}%
    \special{pn 8}%
    \special{pa 224 70}%
    \special{pa 76 159}%
    \special{pa 76 276}%
    \special{pa 429 276}%
    \special{pa 429 159}%
    \special{pa 291 75}%
    \special{sp}%
    \graphtemp=.5ex
    \advance\graphtemp by 0.206in
    \rlap{\kern 0.253in\lower\graphtemp\hbox to 0pt{\hss $a$\hss}}%
    \special{sh 0.500}%
    \special{ar 841 41 41 41 0 6.28319}%
    \special{sh 1.000}%
    \special{pn 1}%
    \special{pa 700 16}%
    \special{pa 800 41}%
    \special{pa 700 66}%
    \special{pa 700 16}%
    \special{fp}%
    \special{pn 8}%
    \special{pa 294 41}%
    \special{pa 700 41}%
    \special{fp}%
    \graphtemp=.5ex
    \advance\graphtemp by 0.118in
    \rlap{\kern 0.565in\lower\graphtemp\hbox to 0pt{\hss $\tau$\hss}}%
    \special{sh 1.000}%
    \special{pn 1}%
    \special{pa 943 143}%
    \special{pa 870 70}%
    \special{pa 969 100}%
    \special{pa 943 143}%
    \special{fp}%
    \special{pn 8}%
    \special{pa 812 70}%
    \special{pa 665 159}%
    \special{pa 665 276}%
    \special{pa 1018 276}%
    \special{pa 1018 159}%
    \special{pa 879 75}%
    \special{sp}%
    \graphtemp=.5ex
    \advance\graphtemp by 0.206in
    \rlap{\kern 0.841in\lower\graphtemp\hbox to 0pt{\hss $a$\hss}}%
    \hbox{\vrule depth0.262in width0pt height 0pt}%
    \kern 0.999in
  }%
}%

%% file: a.tex
\expandafter\ifx\csname graph\endcsname\relax
   \csname newbox\expandafter\endcsname\csname graph\endcsname
\fi
\ifx\graphtemp\undefined
  \csname newdimen\endcsname\graphtemp
\fi
\expandafter\setbox\csname graph\endcsname
 =\vtop{\vskip 0pt\hbox{%
    \special{pn 8}%
    \special{ar 252 400 28 28 0 6.28319}%
    \special{sh 1.000}%
    \special{pn 1}%
    \special{pa 124 375}%
    \special{pa 224 400}%
    \special{pa 124 425}%
    \special{pa 124 375}%
    \special{fp}%
    \special{pn 8}%
    \special{pa 0 400}%
    \special{pa 124 400}%
    \special{fp}%
    \graphtemp=.5ex
    \advance\graphtemp by 0.000in
    \rlap{\kern 0.132in\lower\graphtemp\hbox to 0pt{\hss (a)\hss}}%
    \special{ar 652 400 28 28 0 6.28319}%
    \special{sh 1.000}%
    \special{pn 1}%
    \special{pa 524 375}%
    \special{pa 624 400}%
    \special{pa 524 425}%
    \special{pa 524 375}%
    \special{fp}%
    \special{pn 8}%
    \special{pa 280 400}%
    \special{pa 524 400}%
    \special{fp}%
    \graphtemp=\baselineskip
    \multiply\graphtemp by 1
    \divide\graphtemp by 2
    \advance\graphtemp by .5ex
    \advance\graphtemp by 0.400in
    \rlap{\kern 0.452in\lower\graphtemp\hbox to 0pt{\hss $i$\hss}}%
    \special{ar 972 80 28 28 0 6.28319}%
    \special{sh 1.000}%
    \special{pn 1}%
    \special{pa 864 153}%
    \special{pa 952 100}%
    \special{pa 899 188}%
    \special{pa 864 153}%
    \special{fp}%
    \special{pn 8}%
    \special{pa 672 380}%
    \special{pa 881 171}%
    \special{fp}%
    \graphtemp=\baselineskip
    \multiply\graphtemp by -1
    \divide\graphtemp by 2
    \advance\graphtemp by .5ex
    \advance\graphtemp by 0.240in
    \rlap{\kern 0.812in\lower\graphtemp\hbox to 0pt{\hss $r_1~~$\hss}}%
    \special{ar 972 720 28 28 0 6.28319}%
    \special{sh 1.000}%
    \special{pn 1}%
    \special{pa 899 612}%
    \special{pa 952 700}%
    \special{pa 864 647}%
    \special{pa 899 612}%
    \special{fp}%
    \special{pn 8}%
    \special{pa 672 420}%
    \special{pa 881 629}%
    \special{fp}%
    \graphtemp=\baselineskip
    \multiply\graphtemp by 1
    \divide\graphtemp by 2
    \advance\graphtemp by .5ex
    \advance\graphtemp by 0.560in
    \rlap{\kern 0.812in\lower\graphtemp\hbox to 0pt{\hss $r_0~~$\hss}}%
    \special{ar 1292 400 28 28 0 6.28319}%
    \special{sh 1.000}%
    \special{pn 1}%
    \special{pa 1184 473}%
    \special{pa 1272 420}%
    \special{pa 1219 508}%
    \special{pa 1184 473}%
    \special{fp}%
    \special{pn 8}%
    \special{pa 992 700}%
    \special{pa 1201 491}%
    \special{fp}%
    \graphtemp=\baselineskip
    \multiply\graphtemp by 1
    \divide\graphtemp by 2
    \advance\graphtemp by .5ex
    \advance\graphtemp by 0.560in
    \rlap{\kern 1.132in\lower\graphtemp\hbox to 0pt{\hss $~~s_0$\hss}}%
    \special{sh 1.000}%
    \special{pn 1}%
    \special{pa 1219 292}%
    \special{pa 1272 380}%
    \special{pa 1184 327}%
    \special{pa 1219 292}%
    \special{fp}%
    \special{pn 8}%
    \special{pa 992 100}%
    \special{pa 1201 309}%
    \special{fp}%
    \graphtemp=\baselineskip
    \multiply\graphtemp by -1
    \divide\graphtemp by 2
    \advance\graphtemp by .5ex
    \advance\graphtemp by 0.240in
    \rlap{\kern 1.132in\lower\graphtemp\hbox to 0pt{\hss $~~s_1$\hss}}%
    \special{sh 1.000}%
    \special{pn 1}%
    \special{pa 350 313}%
    \special{pa 272 380}%
    \special{pa 310 284}%
    \special{pa 350 313}%
    \special{fp}%
    \special{pn 8}%
    \special{pa 1320 400}%
    \special{pa 1372 400}%
    \special{pa 1492 280}%
    \special{pa 1492 -39}%
    \special{pa 572 -39}%
    \special{pa 278 372}%
    \special{sp}%
    \graphtemp=.5ex
    \advance\graphtemp by 0.400in
    \rlap{\kern 1.492in\lower\graphtemp\hbox to 0pt{\hss $\tau$\hss}}%
    \hbox{\vrule depth0.748in width0pt height 0pt}%
    \kern 1.492in
  }%
}%

%% file: b.tex
\expandafter\ifx\csname graph\endcsname\relax
   \csname newbox\expandafter\endcsname\csname graph\endcsname
\fi
\ifx\graphtemp\undefined
  \csname newdimen\endcsname\graphtemp
\fi
\expandafter\setbox\csname graph\endcsname
 =\vtop{\vskip 0pt\hbox{%
    \special{pn 8}%
    \special{ar 252 400 28 28 0 6.28319}%
    \special{sh 1.000}%
    \special{pn 1}%
    \special{pa 124 375}%
    \special{pa 224 400}%
    \special{pa 124 425}%
    \special{pa 124 375}%
    \special{fp}%
    \special{pn 8}%
    \special{pa 0 400}%
    \special{pa 124 400}%
    \special{fp}%
    \graphtemp=.5ex
    \advance\graphtemp by 0.000in
    \rlap{\kern 0.132in\lower\graphtemp\hbox to 0pt{\hss (b)\hss}}%
    \special{ar 652 400 28 28 0 6.28319}%
    \special{sh 1.000}%
    \special{pn 1}%
    \special{pa 524 375}%
    \special{pa 624 400}%
    \special{pa 524 425}%
    \special{pa 524 375}%
    \special{fp}%
    \special{pn 8}%
    \special{pa 280 400}%
    \special{pa 524 400}%
    \special{fp}%
    \graphtemp=\baselineskip
    \multiply\graphtemp by 1
    \divide\graphtemp by 2
    \advance\graphtemp by .5ex
    \advance\graphtemp by 0.400in
    \rlap{\kern 0.452in\lower\graphtemp\hbox to 0pt{\hss $\tau$\hss}}%
    \special{ar 972 80 28 28 0 6.28319}%
    \special{sh 1.000}%
    \special{pn 1}%
    \special{pa 864 153}%
    \special{pa 952 100}%
    \special{pa 899 188}%
    \special{pa 864 153}%
    \special{fp}%
    \special{pn 8}%
    \special{pa 672 380}%
    \special{pa 881 171}%
    \special{fp}%
    \graphtemp=\baselineskip
    \multiply\graphtemp by -1
    \divide\graphtemp by 2
    \advance\graphtemp by .5ex
    \advance\graphtemp by 0.240in
    \rlap{\kern 0.812in\lower\graphtemp\hbox to 0pt{\hss $r_1~~$\hss}}%
    \special{ar 972 720 28 28 0 6.28319}%
    \special{sh 1.000}%
    \special{pn 1}%
    \special{pa 899 612}%
    \special{pa 952 700}%
    \special{pa 864 647}%
    \special{pa 899 612}%
    \special{fp}%
    \special{pn 8}%
    \special{pa 672 420}%
    \special{pa 881 629}%
    \special{fp}%
    \graphtemp=\baselineskip
    \multiply\graphtemp by 1
    \divide\graphtemp by 2
    \advance\graphtemp by .5ex
    \advance\graphtemp by 0.560in
    \rlap{\kern 0.812in\lower\graphtemp\hbox to 0pt{\hss $\tau~~$\hss}}%
    \special{ar 1292 400 28 28 0 6.28319}%
    \special{sh 1.000}%
    \special{pn 1}%
    \special{pa 1184 473}%
    \special{pa 1272 420}%
    \special{pa 1219 508}%
    \special{pa 1184 473}%
    \special{fp}%
    \special{pn 8}%
    \special{pa 992 700}%
    \special{pa 1201 491}%
    \special{fp}%
    \graphtemp=\baselineskip
    \multiply\graphtemp by 1
    \divide\graphtemp by 2
    \advance\graphtemp by .5ex
    \advance\graphtemp by 0.560in
    \rlap{\kern 1.132in\lower\graphtemp\hbox to 0pt{\hss $~~\tau$\hss}}%
    \special{sh 1.000}%
    \special{pn 1}%
    \special{pa 1219 292}%
    \special{pa 1272 380}%
    \special{pa 1184 327}%
    \special{pa 1219 292}%
    \special{fp}%
    \special{pn 8}%
    \special{pa 992 100}%
    \special{pa 1201 309}%
    \special{fp}%
    \graphtemp=\baselineskip
    \multiply\graphtemp by -1
    \divide\graphtemp by 2
    \advance\graphtemp by .5ex
    \advance\graphtemp by 0.240in
    \rlap{\kern 1.132in\lower\graphtemp\hbox to 0pt{\hss $~~s_1$\hss}}%
    \special{sh 1.000}%
    \special{pn 1}%
    \special{pa 350 313}%
    \special{pa 272 380}%
    \special{pa 310 284}%
    \special{pa 350 313}%
    \special{fp}%
    \special{pn 8}%
    \special{pa 1320 400}%
    \special{pa 1372 400}%
    \special{pa 1492 280}%
    \special{pa 1492 -39}%
    \special{pa 572 -39}%
    \special{pa 278 372}%
    \special{sp}%
    \graphtemp=.5ex
    \advance\graphtemp by 0.400in
    \rlap{\kern 1.492in\lower\graphtemp\hbox to 0pt{\hss $\tau$\hss}}%
    \hbox{\vrule depth0.748in width0pt height 0pt}%
    \kern 1.492in
  }%
}%

%% file: d.tex
\expandafter\ifx\csname graph\endcsname\relax
   \csname newbox\expandafter\endcsname\csname graph\endcsname
\fi
\ifx\graphtemp\undefined
  \csname newdimen\endcsname\graphtemp
\fi
\expandafter\setbox\csname graph\endcsname
 =\vtop{\vskip 0pt\hbox{%
    \special{pn 8}%
    \special{ar 252 400 28 28 0 6.28319}%
    \special{sh 1.000}%
    \special{pn 1}%
    \special{pa 124 375}%
    \special{pa 224 400}%
    \special{pa 124 425}%
    \special{pa 124 375}%
    \special{fp}%
    \special{pn 8}%
    \special{pa 0 400}%
    \special{pa 124 400}%
    \special{fp}%
    \graphtemp=.5ex
    \advance\graphtemp by 0.000in
    \rlap{\kern 0.132in\lower\graphtemp\hbox to 0pt{\hss (d)\hss}}%
    \special{ar 652 400 28 28 0 6.28319}%
    \special{sh 1.000}%
    \special{pn 1}%
    \special{pa 524 375}%
    \special{pa 624 400}%
    \special{pa 524 425}%
    \special{pa 524 375}%
    \special{fp}%
    \special{pn 8}%
    \special{pa 280 400}%
    \special{pa 524 400}%
    \special{fp}%
    \graphtemp=\baselineskip
    \multiply\graphtemp by 1
    \divide\graphtemp by 2
    \advance\graphtemp by .5ex
    \advance\graphtemp by 0.400in
    \rlap{\kern 0.452in\lower\graphtemp\hbox to 0pt{\hss $i$\hss}}%
    \special{ar 972 80 28 28 0 6.28319}%
    \special{sh 1.000}%
    \special{pn 1}%
    \special{pa 864 153}%
    \special{pa 952 100}%
    \special{pa 899 188}%
    \special{pa 864 153}%
    \special{fp}%
    \special{pn 8}%
    \special{pa 672 380}%
    \special{pa 881 171}%
    \special{fp}%
    \graphtemp=\baselineskip
    \multiply\graphtemp by -1
    \divide\graphtemp by 2
    \advance\graphtemp by .5ex
    \advance\graphtemp by 0.240in
    \rlap{\kern 0.812in\lower\graphtemp\hbox to 0pt{\hss $r_1~~$\hss}}%
    \special{ar 972 720 28 28 0 6.28319}%
    \special{sh 1.000}%
    \special{pn 1}%
    \special{pa 899 612}%
    \special{pa 952 700}%
    \special{pa 864 647}%
    \special{pa 899 612}%
    \special{fp}%
    \special{pn 8}%
    \special{pa 672 420}%
    \special{pa 881 629}%
    \special{fp}%
    \graphtemp=\baselineskip
    \multiply\graphtemp by 1
    \divide\graphtemp by 2
    \advance\graphtemp by .5ex
    \advance\graphtemp by 0.560in
    \rlap{\kern 0.812in\lower\graphtemp\hbox to 0pt{\hss $r_0~~$\hss}}%
    \special{ar 1292 400 28 28 0 6.28319}%
    \special{sh 1.000}%
    \special{pn 1}%
    \special{pa 1184 473}%
    \special{pa 1272 420}%
    \special{pa 1219 508}%
    \special{pa 1184 473}%
    \special{fp}%
    \special{pn 8}%
    \special{pa 992 700}%
    \special{pa 1201 491}%
    \special{fp}%
    \graphtemp=\baselineskip
    \multiply\graphtemp by 1
    \divide\graphtemp by 2
    \advance\graphtemp by .5ex
    \advance\graphtemp by 0.560in
    \rlap{\kern 1.132in\lower\graphtemp\hbox to 0pt{\hss $~~s_0$\hss}}%
    \special{sh 1.000}%
    \special{pn 1}%
    \special{pa 1219 292}%
    \special{pa 1272 380}%
    \special{pa 1184 327}%
    \special{pa 1219 292}%
    \special{fp}%
    \special{pn 8}%
    \special{pa 992 100}%
    \special{pa 1201 309}%
    \special{fp}%
    \graphtemp=\baselineskip
    \multiply\graphtemp by -1
    \divide\graphtemp by 2
    \advance\graphtemp by .5ex
    \advance\graphtemp by 0.240in
    \rlap{\kern 1.132in\lower\graphtemp\hbox to 0pt{\hss $~~s_1$\hss}}%
    \special{sh 1.000}%
    \special{pn 1}%
    \special{pa 350 313}%
    \special{pa 272 380}%
    \special{pa 310 284}%
    \special{pa 350 313}%
    \special{fp}%
    \special{pn 8}%
    \special{pa 1320 400}%
    \special{pa 1372 400}%
    \special{pa 1492 280}%
    \special{pa 1492 -39}%
    \special{pa 572 -39}%
    \special{pa 278 372}%
    \special{sp}%
    \graphtemp=.5ex
    \advance\graphtemp by 0.400in
    \rlap{\kern 1.492in\lower\graphtemp\hbox to 0pt{\hss $\tau$\hss}}%
    \special{ar 92 920 28 28 0 6.28319}%
    \special{ar 492 920 28 28 0 6.28319}%
    \special{sh 1.000}%
    \special{pn 1}%
    \special{pa 364 895}%
    \special{pa 464 920}%
    \special{pa 364 945}%
    \special{pa 364 895}%
    \special{fp}%
    \special{pn 8}%
    \special{pa 120 920}%
    \special{pa 364 920}%
    \special{fp}%
    \graphtemp=\baselineskip
    \multiply\graphtemp by 1
    \divide\graphtemp by 2
    \advance\graphtemp by .5ex
    \advance\graphtemp by 0.920in
    \rlap{\kern 0.292in\lower\graphtemp\hbox to 0pt{\hss $i$\hss}}%
    \special{sh 1.000}%
    \special{pn 1}%
    \special{pa 463 793}%
    \special{pa 481 894}%
    \special{pa 418 813}%
    \special{pa 463 793}%
    \special{fp}%
    \special{pn 8}%
    \special{pa 283 445}%
    \special{pa 440 803}%
    \special{fp}%
    \graphtemp=\baselineskip
    \multiply\graphtemp by 1
    \divide\graphtemp by 2
    \advance\graphtemp by .5ex
    \advance\graphtemp by 0.670in
    \rlap{\kern 0.382in\lower\graphtemp\hbox to 0pt{\hss $i~~$\hss}}%
    \special{ar 892 920 28 28 0 6.28319}%
    \special{sh 1.000}%
    \special{pn 1}%
    \special{pa 764 895}%
    \special{pa 864 920}%
    \special{pa 764 945}%
    \special{pa 764 895}%
    \special{fp}%
    \special{pn 8}%
    \special{pa 520 920}%
    \special{pa 764 920}%
    \special{fp}%
    \graphtemp=\baselineskip
    \multiply\graphtemp by 1
    \divide\graphtemp by 2
    \advance\graphtemp by .5ex
    \advance\graphtemp by 0.920in
    \rlap{\kern 0.692in\lower\graphtemp\hbox to 0pt{\hss $r_0$\hss}}%
    \special{ar 1292 920 28 28 0 6.28319}%
    \special{sh 1.000}%
    \special{pn 1}%
    \special{pa 1164 895}%
    \special{pa 1264 920}%
    \special{pa 1164 945}%
    \special{pa 1164 895}%
    \special{fp}%
    \special{pn 8}%
    \special{pa 920 920}%
    \special{pa 1164 920}%
    \special{fp}%
    \graphtemp=\baselineskip
    \multiply\graphtemp by 1
    \divide\graphtemp by 2
    \advance\graphtemp by .5ex
    \advance\graphtemp by 0.920in
    \rlap{\kern 1.092in\lower\graphtemp\hbox to 0pt{\hss $s_0$\hss}}%
    \special{sh 1.000}%
    \special{pn 1}%
    \special{pa 150 1036}%
    \special{pa 112 940}%
    \special{pa 190 1007}%
    \special{pa 150 1036}%
    \special{fp}%
    \special{pn 8}%
    \special{pa 1320 920}%
    \special{pa 1372 920}%
    \special{pa 1492 1040}%
    \special{pa 1492 1360}%
    \special{pa 412 1360}%
    \special{pa 118 948}%
    \special{sp}%
    \graphtemp=.5ex
    \advance\graphtemp by 0.920in
    \rlap{\kern 1.492in\lower\graphtemp\hbox to 0pt{\hss $\tau$\hss}}%
    \hbox{\vrule depth1.320in width0pt height 0pt}%
    \kern 1.492in
  }%
}%

%% file: e.tex
\expandafter\ifx\csname graph\endcsname\relax
   \csname newbox\expandafter\endcsname\csname graph\endcsname
\fi
\ifx\graphtemp\undefined
  \csname newdimen\endcsname\graphtemp
\fi
\expandafter\setbox\csname graph\endcsname
 =\vtop{\vskip 0pt\hbox{%
    \special{pn 8}%
    \special{ar 252 400 28 28 0 6.28319}%
    \special{sh 1.000}%
    \special{pn 1}%
    \special{pa 124 375}%
    \special{pa 224 400}%
    \special{pa 124 425}%
    \special{pa 124 375}%
    \special{fp}%
    \special{pn 8}%
    \special{pa 0 400}%
    \special{pa 124 400}%
    \special{fp}%
    \graphtemp=.5ex
    \advance\graphtemp by 0.000in
    \rlap{\kern 0.132in\lower\graphtemp\hbox to 0pt{\hss (e)\hss}}%
    \special{ar 652 400 28 28 0 6.28319}%
    \special{sh 1.000}%
    \special{pn 1}%
    \special{pa 524 375}%
    \special{pa 624 400}%
    \special{pa 524 425}%
    \special{pa 524 375}%
    \special{fp}%
    \special{pn 8}%
    \special{pa 280 400}%
    \special{pa 524 400}%
    \special{fp}%
    \graphtemp=\baselineskip
    \multiply\graphtemp by 1
    \divide\graphtemp by 2
    \advance\graphtemp by .5ex
    \advance\graphtemp by 0.400in
    \rlap{\kern 0.452in\lower\graphtemp\hbox to 0pt{\hss $i$\hss}}%
    \special{ar 972 80 28 28 0 6.28319}%
    \special{sh 1.000}%
    \special{pn 1}%
    \special{pa 864 153}%
    \special{pa 952 100}%
    \special{pa 899 188}%
    \special{pa 864 153}%
    \special{fp}%
    \special{pn 8}%
    \special{pa 672 380}%
    \special{pa 881 171}%
    \special{fp}%
    \graphtemp=\baselineskip
    \multiply\graphtemp by -1
    \divide\graphtemp by 2
    \advance\graphtemp by .5ex
    \advance\graphtemp by 0.240in
    \rlap{\kern 0.812in\lower\graphtemp\hbox to 0pt{\hss $\tau~~$\hss}}%
    \special{ar 972 720 28 28 0 6.28319}%
    \special{sh 1.000}%
    \special{pn 1}%
    \special{pa 899 612}%
    \special{pa 952 700}%
    \special{pa 864 647}%
    \special{pa 899 612}%
    \special{fp}%
    \special{pn 8}%
    \special{pa 672 420}%
    \special{pa 881 629}%
    \special{fp}%
    \graphtemp=\baselineskip
    \multiply\graphtemp by 1
    \divide\graphtemp by 2
    \advance\graphtemp by .5ex
    \advance\graphtemp by 0.560in
    \rlap{\kern 0.812in\lower\graphtemp\hbox to 0pt{\hss $\tau~~$\hss}}%
    \special{ar 1292 400 28 28 0 6.28319}%
    \special{sh 1.000}%
    \special{pn 1}%
    \special{pa 1184 473}%
    \special{pa 1272 420}%
    \special{pa 1219 508}%
    \special{pa 1184 473}%
    \special{fp}%
    \special{pn 8}%
    \special{pa 992 700}%
    \special{pa 1201 491}%
    \special{fp}%
    \graphtemp=\baselineskip
    \multiply\graphtemp by 1
    \divide\graphtemp by 2
    \advance\graphtemp by .5ex
    \advance\graphtemp by 0.560in
    \rlap{\kern 1.132in\lower\graphtemp\hbox to 0pt{\hss $~~s_0$\hss}}%
    \special{sh 1.000}%
    \special{pn 1}%
    \special{pa 1219 292}%
    \special{pa 1272 380}%
    \special{pa 1184 327}%
    \special{pa 1219 292}%
    \special{fp}%
    \special{pn 8}%
    \special{pa 992 100}%
    \special{pa 1201 309}%
    \special{fp}%
    \graphtemp=\baselineskip
    \multiply\graphtemp by -1
    \divide\graphtemp by 2
    \advance\graphtemp by .5ex
    \advance\graphtemp by 0.240in
    \rlap{\kern 1.132in\lower\graphtemp\hbox to 0pt{\hss $~~s_1$\hss}}%
    \special{sh 1.000}%
    \special{pn 1}%
    \special{pa 350 313}%
    \special{pa 272 380}%
    \special{pa 310 284}%
    \special{pa 350 313}%
    \special{fp}%
    \special{pn 8}%
    \special{pa 1320 400}%
    \special{pa 1372 400}%
    \special{pa 1492 280}%
    \special{pa 1492 -39}%
    \special{pa 572 -39}%
    \special{pa 278 372}%
    \special{sp}%
    \graphtemp=.5ex
    \advance\graphtemp by 0.400in
    \rlap{\kern 1.492in\lower\graphtemp\hbox to 0pt{\hss $\tau$\hss}}%
    \hbox{\vrule depth0.748in width0pt height 0pt}%
    \kern 1.492in
  }%
}%

%% file: c.tex
\expandafter\ifx\csname graph\endcsname\relax
   \csname newbox\expandafter\endcsname\csname graph\endcsname
\fi
\ifx\graphtemp\undefined
  \csname newdimen\endcsname\graphtemp
\fi
\expandafter\setbox\csname graph\endcsname
 =\vtop{\vskip 0pt\hbox{%
    \special{pn 8}%
    \special{ar 252 400 28 28 0 6.28319}%
    \special{sh 1.000}%
    \special{pn 1}%
    \special{pa 124 375}%
    \special{pa 224 400}%
    \special{pa 124 425}%
    \special{pa 124 375}%
    \special{fp}%
    \special{pn 8}%
    \special{pa 0 400}%
    \special{pa 124 400}%
    \special{fp}%
    \graphtemp=.5ex
    \advance\graphtemp by 0.000in
    \rlap{\kern 0.132in\lower\graphtemp\hbox to 0pt{\hss (c)\hss}}%
    \special{ar 652 400 28 28 0 6.28319}%
    \special{sh 1.000}%
    \special{pn 1}%
    \special{pa 524 375}%
    \special{pa 624 400}%
    \special{pa 524 425}%
    \special{pa 524 375}%
    \special{fp}%
    \special{pn 8}%
    \special{pa 280 400}%
    \special{pa 524 400}%
    \special{fp}%
    \graphtemp=\baselineskip
    \multiply\graphtemp by 1
    \divide\graphtemp by 2
    \advance\graphtemp by .5ex
    \advance\graphtemp by 0.400in
    \rlap{\kern 0.452in\lower\graphtemp\hbox to 0pt{\hss $\tau$\hss}}%
    \special{ar 972 80 28 28 0 6.28319}%
    \special{sh 1.000}%
    \special{pn 1}%
    \special{pa 864 153}%
    \special{pa 952 100}%
    \special{pa 899 188}%
    \special{pa 864 153}%
    \special{fp}%
    \special{pn 8}%
    \special{pa 672 380}%
    \special{pa 881 171}%
    \special{fp}%
    \graphtemp=\baselineskip
    \multiply\graphtemp by -1
    \divide\graphtemp by 2
    \advance\graphtemp by .5ex
    \advance\graphtemp by 0.240in
    \rlap{\kern 0.812in\lower\graphtemp\hbox to 0pt{\hss $r_1~~$\hss}}%
    \special{ar 1292 400 28 28 0 6.28319}%
    \special{sh 1.000}%
    \special{pn 1}%
    \special{pa 1219 292}%
    \special{pa 1272 380}%
    \special{pa 1184 327}%
    \special{pa 1219 292}%
    \special{fp}%
    \special{pn 8}%
    \special{pa 992 100}%
    \special{pa 1201 309}%
    \special{fp}%
    \graphtemp=\baselineskip
    \multiply\graphtemp by -1
    \divide\graphtemp by 2
    \advance\graphtemp by .5ex
    \advance\graphtemp by 0.240in
    \rlap{\kern 1.132in\lower\graphtemp\hbox to 0pt{\hss $~~s_1$\hss}}%
    \special{sh 1.000}%
    \special{pn 1}%
    \special{pa 350 313}%
    \special{pa 272 380}%
    \special{pa 310 284}%
    \special{pa 350 313}%
    \special{fp}%
    \special{pn 8}%
    \special{pa 1320 400}%
    \special{pa 1372 400}%
    \special{pa 1492 280}%
    \special{pa 1492 -39}%
    \special{pa 572 -39}%
    \special{pa 278 372}%
    \special{sp}%
    \graphtemp=.5ex
    \advance\graphtemp by 0.400in
    \rlap{\kern 1.492in\lower\graphtemp\hbox to 0pt{\hss $\tau$\hss}}%
    \hbox{\vrule depth0.428in width0pt height 0pt}%
    \kern 1.492in
  }%
}%

%% file: f.tex
\expandafter\ifx\csname graph\endcsname\relax
   \csname newbox\expandafter\endcsname\csname graph\endcsname
\fi
\ifx\graphtemp\undefined
  \csname newdimen\endcsname\graphtemp
\fi
\expandafter\setbox\csname graph\endcsname
 =\vtop{\vskip 0pt\hbox{%
    \hbox{\vrule depth0.748in width0pt height 0pt}%
    \kern 1.477in
  }%
}%

%% file: Cataline.tex
\expandafter\ifx\csname graph\endcsname\relax
   \csname newbox\expandafter\endcsname\csname graph\endcsname
\fi
\ifx\graphtemp\undefined
  \csname newdimen\endcsname\graphtemp
\fi
\expandafter\setbox\csname graph\endcsname
 =\vtop{\vskip 0pt\hbox{%
    \special{pn 8}%
    \special{ar 330 70 70 70 0 6.28319}%
    \graphtemp=.5ex
    \advance\graphtemp by 0.070in
    \rlap{\kern 0.330in\lower\graphtemp\hbox to 0pt{\hss {\scriptsize 1}\hss}}%
    \special{sh 1.000}%
    \special{pn 1}%
    \special{pa 160 45}%
    \special{pa 260 70}%
    \special{pa 160 95}%
    \special{pa 160 45}%
    \special{fp}%
    \special{pn 8}%
    \special{pa 0 70}%
    \special{pa 160 70}%
    \special{fp}%
    \special{sh 0.300}%
    \special{ar 1330 70 70 70 0 6.28319}%
    \graphtemp=.5ex
    \advance\graphtemp by 0.070in
    \rlap{\kern 1.330in\lower\graphtemp\hbox to 0pt{\hss {\scriptsize 2}\hss}}%
    \special{sh 1.000}%
    \special{pn 1}%
    \special{pa 1160 45}%
    \special{pa 1260 70}%
    \special{pa 1160 95}%
    \special{pa 1160 45}%
    \special{fp}%
    \special{pn 8}%
    \special{pa 400 70}%
    \special{pa 1160 70}%
    \special{fp}%
    \graphtemp=\baselineskip
    \multiply\graphtemp by -1
    \divide\graphtemp by 2
    \advance\graphtemp by .5ex
    \advance\graphtemp by 0.070in
    \rlap{\kern 0.830in\lower\graphtemp\hbox to 0pt{\hss {\it cr}\hss}}%
    \hbox{\vrule depth0.140in width0pt height 0pt}%
    \kern 1.400in
  }%
}%

%% file: Alice.tex
\expandafter\ifx\csname graph\endcsname\relax
   \csname newbox\expandafter\endcsname\csname graph\endcsname
\fi
\ifx\graphtemp\undefined
  \csname newdimen\endcsname\graphtemp
\fi
\expandafter\setbox\csname graph\endcsname
 =\vtop{\vskip 0pt\hbox{%
    \special{pn 8}%
    \special{ar 253 41 41 41 0 6.28319}%
    \special{sh 1.000}%
    \special{pn 1}%
    \special{pa 112 16}%
    \special{pa 212 41}%
    \special{pa 112 66}%
    \special{pa 112 16}%
    \special{fp}%
    \special{pn 8}%
    \special{pa 0 41}%
    \special{pa 112 41}%
    \special{fp}%
    \special{sh 1.000}%
    \special{pn 1}%
    \special{pa 355 143}%
    \special{pa 282 70}%
    \special{pa 381 100}%
    \special{pa 355 143}%
    \special{fp}%
    \special{pn 8}%
    \special{pa 224 70}%
    \special{pa 76 159}%
    \special{pa 76 276}%
    \special{pa 429 276}%
    \special{pa 429 159}%
    \special{pa 291 75}%
    \special{sp}%
    \graphtemp=.5ex
    \advance\graphtemp by 0.206in
    \rlap{\kern 0.253in\lower\graphtemp\hbox to 0pt{\hss $a$\hss}}%
    \hbox{\vrule depth0.262in width0pt height 0pt}%
    \kern 0.411in
  }%
}%

%% file: AliceCataline.tex
\expandafter\ifx\csname graph\endcsname\relax
   \csname newbox\expandafter\endcsname\csname graph\endcsname
\fi
\ifx\graphtemp\undefined
  \csname newdimen\endcsname\graphtemp
\fi
\expandafter\setbox\csname graph\endcsname
 =\vtop{\vskip 0pt\hbox{%
    \special{pn 8}%
    \special{ar 253 76 41 41 0 6.28319}%
    \special{sh 1.000}%
    \special{pn 1}%
    \special{pa 112 51}%
    \special{pa 212 76}%
    \special{pa 112 101}%
    \special{pa 112 51}%
    \special{fp}%
    \special{pn 8}%
    \special{pa 0 76}%
    \special{pa 112 76}%
    \special{fp}%
    \special{sh 1.000}%
    \special{pn 1}%
    \special{pa 355 179}%
    \special{pa 282 106}%
    \special{pa 381 136}%
    \special{pa 355 179}%
    \special{fp}%
    \special{pn 8}%
    \special{pa 224 106}%
    \special{pa 76 194}%
    \special{pa 76 312}%
    \special{pa 429 312}%
    \special{pa 429 194}%
    \special{pa 291 111}%
    \special{sp}%
    \graphtemp=.5ex
    \advance\graphtemp by 0.241in
    \rlap{\kern 0.253in\lower\graphtemp\hbox to 0pt{\hss $a$\hss}}%
    \special{sh 0.500}%
    \special{ar 841 76 41 41 0 6.28319}%
    \special{sh 1.000}%
    \special{pn 1}%
    \special{pa 700 51}%
    \special{pa 800 76}%
    \special{pa 700 101}%
    \special{pa 700 51}%
    \special{fp}%
    \special{pn 8}%
    \special{pa 294 76}%
    \special{pa 700 76}%
    \special{fp}%
    \graphtemp=.5ex
    \advance\graphtemp by 0.000in
    \rlap{\kern 0.565in\lower\graphtemp\hbox to 0pt{\hss {\it cr}\hss}}%
    \special{sh 1.000}%
    \special{pn 1}%
    \special{pa 943 179}%
    \special{pa 870 106}%
    \special{pa 969 136}%
    \special{pa 943 179}%
    \special{fp}%
    \special{pn 8}%
    \special{pa 812 106}%
    \special{pa 665 194}%
    \special{pa 665 312}%
    \special{pa 1018 312}%
    \special{pa 1018 194}%
    \special{pa 879 111}%
    \special{sp}%
    \graphtemp=.5ex
    \advance\graphtemp by 0.241in
    \rlap{\kern 0.841in\lower\graphtemp\hbox to 0pt{\hss $a$\hss}}%
    \hbox{\vrule depth0.297in width0pt height 0pt}%
    \kern 0.999in
  }%
}%

%% file: hierarchy.tex
\expandafter\ifx\csname graph\endcsname\relax
   \csname newbox\expandafter\endcsname\csname graph\endcsname
\fi
\ifx\graphtemp\undefined
  \csname newdimen\endcsname\graphtemp
\fi
\expandafter\setbox\csname graph\endcsname
 =\vtop{\vskip 0pt\hbox{%
    \graphtemp=.5ex
    \advance\graphtemp by 2.362in
    \rlap{\kern 0.000in\lower\graphtemp\hbox to 0pt{\hss $\emptyset$\hss}}%
    \graphtemp=.5ex
    \advance\graphtemp by 1.969in
    \rlap{\kern 0.000in\lower\graphtemp\hbox to 0pt{\hss progress\hss}}%
    \graphtemp=.5ex
    \advance\graphtemp by 1.575in
    \rlap{\kern 0.000in\lower\graphtemp\hbox to 0pt{\hss justness\hss}}%
    \graphtemp=.5ex
    \advance\graphtemp by 1.181in
    \rlap{\kern 0.000in\lower\graphtemp\hbox to 0pt{\hss \J-fairness of componenets\hss}}%
    \graphtemp=.5ex
    \advance\graphtemp by 0.787in
    \rlap{\kern 0.000in\lower\graphtemp\hbox to 0pt{\hss weak fairness of componenets\hss}}%
    \graphtemp=.5ex
    \advance\graphtemp by 0.394in
    \rlap{\kern 0.000in\lower\graphtemp\hbox to 0pt{\hss strong fairness of componenets\hss}}%
    \graphtemp=.5ex
    \advance\graphtemp by 0.000in
    \rlap{\kern 0.000in\lower\graphtemp\hbox to 0pt{\hss full fairness\hss}}%
    \graphtemp=.5ex
    \advance\graphtemp by 2.165in
    \rlap{\kern 0.000in\lower\graphtemp\hbox to 0pt{\hss $|$\hss}}%
    \graphtemp=.5ex
    \advance\graphtemp by 1.772in
    \rlap{\kern 0.000in\lower\graphtemp\hbox to 0pt{\hss $|$\hss}}%
    \graphtemp=.5ex
    \advance\graphtemp by 1.378in
    \rlap{\kern 0.000in\lower\graphtemp\hbox to 0pt{\hss $|$\hss}}%
    \graphtemp=.5ex
    \advance\graphtemp by 0.984in
    \rlap{\kern 0.000in\lower\graphtemp\hbox to 0pt{\hss $|$\hss}}%
    \graphtemp=.5ex
    \advance\graphtemp by 0.591in
    \rlap{\kern 0.000in\lower\graphtemp\hbox to 0pt{\hss $|$\hss}}%
    \graphtemp=.5ex
    \advance\graphtemp by 0.197in
    \rlap{\kern 0.000in\lower\graphtemp\hbox to 0pt{\hss $|$\hss}}%
    \hbox{\vrule depth2.362in width0pt height 0pt}%
    \kern 0.000in
  }%
}%

%% file: PQ.tex
\expandafter\ifx\csname graph\endcsname\relax
   \csname newbox\expandafter\endcsname\csname graph\endcsname
\fi
\ifx\graphtemp\undefined
  \csname newdimen\endcsname\graphtemp
\fi
\expandafter\setbox\csname graph\endcsname
 =\vtop{\vskip 0pt\hbox{%
    \special{pn 8}%
    \special{ar 2000 450 50 50 0 6.28319}%
    \special{sh 1.000}%
    \special{pn 1}%
    \special{pa 2118 353}%
    \special{pa 2035 415}%
    \special{pa 2079 321}%
    \special{pa 2118 353}%
    \special{fp}%
    \special{pn 8}%
    \special{pa 2035 485}%
    \special{pa 2250 750}%
    \special{pa 2550 750}%
    \special{pa 2550 150}%
    \special{pa 2250 150}%
    \special{pa 2042 407}%
    \special{sp}%
    \graphtemp=.5ex
    \advance\graphtemp by 0.450in
    \rlap{\kern 2.930in\lower\graphtemp\hbox to 0pt{\hss $y:=y+1$\hss}}%
    \special{ar 1000 450 50 50 0 6.28319}%
    \special{sh 1.000}%
    \special{pn 1}%
    \special{pa 1025 300}%
    \special{pa 1000 400}%
    \special{pa 975 300}%
    \special{pa 1025 300}%
    \special{fp}%
    \special{pn 8}%
    \special{pa 1000 0}%
    \special{pa 1000 300}%
    \special{fp}%
    \special{sh 1.000}%
    \special{pn 1}%
    \special{pa 1850 425}%
    \special{pa 1950 450}%
    \special{pa 1850 475}%
    \special{pa 1850 425}%
    \special{fp}%
    \special{pn 8}%
    \special{pa 1050 450}%
    \special{pa 1850 450}%
    \special{fp}%
    \graphtemp=\baselineskip
    \multiply\graphtemp by -1
    \divide\graphtemp by 2
    \advance\graphtemp by .5ex
    \advance\graphtemp by 0.450in
    \rlap{\kern 1.500in\lower\graphtemp\hbox to 0pt{\hss $x:=1$\hss}}%
    \special{sh 1.000}%
    \special{pn 1}%
    \special{pa 904 331}%
    \special{pa 965 415}%
    \special{pa 872 369}%
    \special{pa 904 331}%
    \special{fp}%
    \special{pn 8}%
    \special{pa 965 485}%
    \special{pa 650 750}%
    \special{pa 350 750}%
    \special{pa 350 150}%
    \special{pa 650 150}%
    \special{pa 957 408}%
    \special{sp}%
    \graphtemp=.5ex
    \advance\graphtemp by 0.450in
    \rlap{\kern 0.000in\lower\graphtemp\hbox to 0pt{\hss $y:=y+1$\hss}}%
    \hbox{\vrule depth0.717in width0pt height 0pt}%
    \kern 2.930in
  }%
}%